\documentclass[a4paper,onecolumn,10pt]{quantumarticle}
\pdfoutput=1
\usepackage[utf8]{inputenc}
\usepackage[english]{babel}
\usepackage[T1]{fontenc}

\usepackage[affil-it]{authblk}
\usepackage[usenames,dvipsnames]{xcolor}
\usepackage{amsfonts}
\usepackage{amsmath,amsthm,amssymb}
\usepackage{enumerate}
\usepackage[english]{babel}
\usepackage{graphicx}	
\usepackage[caption=false]{subfig}
\usepackage{url}
\usepackage{bbm}
\usepackage{blkarray}
\usepackage{braket}
\usepackage{parskip}

\usepackage{tikz,pgfplotstable}
\usetikzlibrary{chains}
\usetikzlibrary{fit}
\usetikzlibrary{spy}
\usepackage{pgflibraryarrows}		
\usepackage{pgflibrarysnakes}		

\usepackage{makecell}

\usepackage{epsfig}
\usetikzlibrary{shapes.symbols,patterns} 
\usepackage{pgfplots}

\usepackage{mathrsfs}

\usepackage{hyperref}
\hypersetup{colorlinks=true,citecolor=blue,linkcolor=blue,filecolor=blue,urlcolor=blue,breaklinks=true}

\usepackage{nicefrac}
\usepackage{mathtools}
\usepackage{xparse}
\usepackage{subfig}

\theoremstyle{plain}
\newtheorem{theorem}{Theorem}[section]
\newtheorem{lemma}[theorem]{Lemma}

\theoremstyle{definition}
\newtheorem{definition}[theorem]{Definition}
\newtheorem{remark}[theorem]{Remark}
\newtheorem{example}[theorem]{Example}

\newcommand*{\cE}{\mathcal{E}}
\newcommand*{\cF}{\mathcal{F}}

\newcommand*{\cI}{\mathcal{I}}

\newcommand*{\cR}{\mathcal{R}}

\newcommand*{\NN}{\mathbb{N}}

\newcommand*{\RR}{\mathbb{R}}

\newcommand*{\CC}{\mathbb{C}}

\newcommand*{\HH}{\mathbb{H}}

\newcommand*{\id}{I}
\newcommand*{\supp}{\mathrm{supp}}

\newcommand*{\Pos}{\mathscr{P}}

\newcommand*{\D}{\mathscr{D}}

\newcommand*{\ox}{\otimes}

\usepackage[normalem]{ulem}

\newcommand{\pH}[1]{\overline{H}_{#1}}

\newcommand{\pHd}[1]{\pH{#1}^{\downarrow}}
\newcommand{\swD}[1]{D_{#1}}  	
\newcommand{\swH}[1]{H_{#1}} 
\newcommand{\swQ}[1]{Q_{#1}}
\newcommand{\swHu}[1]{\swH{#1}^\uparrow}

\newcommand{\swHd}[1]{\swH{#1}^\downarrow}

\newcommand{\tr}[1]{\mathrm{Tr}\left[#1\right]} 
\newcommand{\ptr}[2]{\mathrm{Tr}_{#1}\left[#2\right]}

\newcommand{\ketbra}[2]{\ensuremath{\left|#1\right\rangle\!\left\langle#2\right|}}

\newcommand*{\defed}{:=} 

\newcommand{\maxintrinsicpvm}[2]{\cR_{#1}^{\mathrm{PVM}}\left(#2\right)}
\newcommand{\maxintrinsicpovm}[2]{\cR_{#1}^{\mathrm{POVM}}\left(#2\right)}

\newcommand{\intrinsicpovm}[3]{\cI_{#1}^{\mathrm{POVM}}\left(#2,#3\right)}

\renewcommand{\epsilon}{\varepsilon}

\let\inner\relax
\NewDocumentCommand\inner{mg}{%
	\ensuremath{\left\langle #1| \IfNoValueTF{#2}{#1}{#2}\right\rangle}%
}
\let\outer\relax
\NewDocumentCommand\outer{mg}{%
	\ensuremath{\ket{#1}\! \IfNoValueTF{#2}{\bra{#1}}{\bra{#2}}}%
}

\begin{document}

	\title{{How much secure randomness is in a quantum state?}}
	\author[1]{Kriss Gutierrez Anco}
	\author[1]{Tristan Nemoz}
	\author[1]{Peter Brown}
	\email{peter.brown@telecom-paris.fr}
	\orcid{0000-0001-9593-0136}
	\affil[1]{\small{Télécom Paris, Inria, LTCI, Institut Polytechnique de Paris, 91120 Palaiseau, France}}
	\maketitle

\begin{abstract}
	How much cryptographically-secure randomness can be extracted from a quantum state? This fundamental question probes the absolute limits of quantum random number generation (QRNG) and yet, despite the technological maturity of QRNGs, it remains unsolved. In this work we consider a general adversarial model that allows for an adversary who has quantum side-information about both the source and the measurement device. Using links between randomness extraction rates and sandwiched Rényi entropies, we provide compact, easy to compute, achievable rates of secure randomness extraction from quantum states. In turn, this provides a simple to evaluate benchmarking tool for the randomness generation rates of QRNG protocols. 
\end{abstract}

\section{Introduction}
The intrinsic randomness of quantum systems lies at the heart of quantum cryptography, enabling fundamental cryptographic primitives such as quantum random number generation (QRNG)~\cite{qrng-review} and quantum key distribution (QKD)~\cite{BB84}. Significant progress has been made on from both the theoretical and experimental perspectives and both QRNG and QKD now exist as commercially viable technologies. Despite this progress, a fundamental question remains unsolved: how much secure randomness can you extract from a quantum state?

Security proofs formalize what we mean by secure randomness with precise mathematical definitions~\cite{PR22}. In turn, these definitions enable us to formally prove that, under the various assumptions of the protocol, the randomness generated by our QRNG is secure. Modern security proof frameworks such as the entropy accumulation theorem~\cite{DFR,MFSR22}, rely on quantifying the amount of randomness produced during the protocol with respect to an adversary. This typically amounts to an optimization of a conditional entropy over states that are consistent with the observations and assumptions of the protocol~\cite{DW05}. For example, in device-independent QRNG we minimize some conditional entropy over all states that are consistent with a given Bell-violation~\cite{ARV,PABGMS}. Whereas for device-dependent cryptography, we optimize over states that are consistent with the observations made during the protocol and the trusted components, e.g. the measurement operators are fully characterized~\cite{CML16}. There also exist intermediary scenarios, so-called semi-device-independence, where other assumptions are placed on the systems, e.g. energy bounds~\cite{HWCGP17}.

Whilst these security proof frameworks allow us to prove that we can extract secure randomness from a source of quantum states, a pertinent question for the design of efficient QRNGs remains: what is the most amount of secure randomness we can extract? Such a quantifier would provide us with the means to assess the resourcefulness of a given source for practical QRNG and would further provide a benchmarking tool for finite-size security proofs of QRNG protocols. 

More concretely, consider the setup depicted in Figure~\ref{fig:setup} (detailed in Section~\ref{sec:intrinsic-rand-povm}), which is an adversarial model introduced in~\cite{FRT13} (see also~\cite{SSA23,DCZM23} for more recent developments). In this setup, we consider an honest agent Alice who wishes to generate some secret randomness. To do so, she measures a trusted quantum state $\rho_A$ using a trusted POVM $\{M_x\}_x$ producing a classical random variable $X$. However, unbeknownst to Alice, an eavesdropper Eve has quantum side information about both Alice's source (which produces $\rho_A$) and Alice's measurement device (which measures $\{M_x\}_x$). We model this general quantum side information by a tripartite state $\rho_{ABE}$ where the system $A$ is held by Alice's source, the system $B$ is held by Alice's measurement device (which it can use to help implement the POVM $\{M_x\}_x$ on $A$) and the system $E$ is held by Eve. The marginal state $\rho_A$ and the POVM $\{M_x\}_x$ are both trusted and known but the systems held by the measurement device and the adversary are unknown and can in principal be selected by Eve to provide her with the most amount of information about $X$. 

We are interested in the maximal amount of cryptographically secure randomness that Alice can extract from her system $A$ in such a general adversarial model. To extract the randomness from her state Alice first measures her system using producing the classical system $X$. This results in a classical-quantum post-measurement state $\rho_{XE}$ which captures the potential side-information that Eve has about $X$. To remove this side-information and extract cryptographically secure randomness, Alice performs a randomness extraction (privacy amplification) protocol on the classical system $X$, producing $k$ bits of $\epsilon$-secure randomness. In this work we ask the following question: if we optimize the choice of POVM $\{M_x\}_x$, how many bits of $\epsilon$-secure randomness can we extract? Or, what's an achievable amount of $\epsilon$-secure randomness that we can extract from a quantum state $\rho_A$ against a general quantum adversary?

\subsection{Contributions}
In this work we derive simple, easy to compute expressions for achievable rates of extraction of cryptographically secure randomness from a quantum state. In particular, if we define the Rényi entropy of system $A$ to be $H_\alpha(A) = \frac{1}{1-\alpha} \log \tr{\rho_A^\alpha}$ for $\alpha \in (0,1) \cup (1, \infty)$ then we can show that Alice can extract up to 
\begin{equation}\label{eq:main-res-intro}
	2\log(d_A) - \min_{\alpha \in (1,2]} \left(H_{\frac{\alpha}{2\alpha-1}}(A) - \frac{\alpha}{\alpha-1} \log(\epsilon) \right)
\end{equation}
bits of $\epsilon$-secure randomness from her quantum state, where $d_A$ is the dimension of the system $A$. In addition, if she restricts herself to projective measurements only then she can extract up to
\begin{equation}
	\log(d_A) - \min_{\alpha \in (1,2]} \left(H_{\frac{\alpha}{2\alpha-1}}(A) - \frac{\alpha}{\alpha-1} \log(\epsilon) \right)
\end{equation}
bits of $\epsilon$-secure randomness. These results give simple to evaluate benchmarks for quantum random number generators (QRNG). In particular, given a QRNG that uses a source producing states $\rho_A$, our work shows that such a QRNG could in principle produce randomness at a rate given by~\eqref{eq:main-res-intro}, setting a benchmark for the performance of QRNGs and their security proofs.

On the technical level, the majority of this work is dedicated to proving that 
\begin{equation}\label{eq:sw-opt-intro}
	\sup_{\{M_x\}_x} \swHu{\alpha}(X|E) = 2 \log(d_A) - H_{\frac{\alpha}{2\alpha-1}}(A)
\end{equation}
where $\swHu{\alpha}$ is the conditional sandwiched Rényi entropy~\cite{MDSFT,WWY13} and $\rho_{XE}$ is some post-measurement state defined by the adversarial model (see Section~\ref{sec:intrinsic-rand-povm} and Figure~\ref{fig:setup} for details). Along the way to proving this, we define what we call the \emph{maximal intrinsic $\HH$-randomness problem} which is the optimization problem
\begin{equation}
	\sup_{\{M_x\}_x} \HH(X|E)
\end{equation}
defined for any conditional entropy $\HH$. This was inspired by a similar problem introduced recently in~\cite{MCSWFSA23}. However, in that work, $\HH$ is restricted to the von Neumann entropy and the min/max entropies, $\{M_x\}_x$ is restricted to projective measurements and no entanglement with the measurement device was considered. Generalizing this problem to other conditional entropies allows us to connect to the operationally relevant notion of extractable $\epsilon$-secure randomness~\cite{D21}. Furthermore, we are able to establish our results for general measurements in a general quantum adversarial model wherein the source, the measurement device and Eve can all be entangled. Note that when we restrict our optimization to projective measurements we are able to recover the von Neumann entropy and min-entropy results of~\cite{MCSWFSA23}, showing that they also hold in this stronger adversarial model.

The work is structured as follows. In Section~\ref{sec:prelims} we introduce the relevant notation and entropic quantities used throughout the work. In Section~\ref{sec:intrinsic-rand-proj} we solve a simpler version of the maximal intrinsic $\HH$-randomness problem where the measurements are restricted to be projective and no entanglement is shared with the measurement device. In Section~\ref{sec:intrinsic-rand-povm} we consider the fully general scenario giving closed form expressions for the maximal intrinsic $\HH$-randomness problem for the operationally relevant conditional sandwiched Rényi entropies. In Section~\ref{sec:application} we connect the solutions of Section~\ref{sec:intrinsic-rand-proj} and Section~\ref{sec:intrinsic-rand-povm} to the problem of maximal extractable secure randomness and give a simple to compute benchmark for the performance of QRNGs. 

\section{Preliminaries}\label{sec:prelims}
We begin by introducing some notation used throughout the article. Quantum systems and their associated Hilbert spaces are denoted by capital Roman letters, e.g. $A$ or $B$. All Hilbert spaces considered in this work are finite dimensional. The set of positive semidefinite square matrices acting on a system $A$ are denoted~$\Pos(A)$. Given two Hermitian matrices $X,Y$ acting on a Hilbert space $A$, we write $X \succeq Y$ if $X-Y \in \Pos(A)$.  A \emph{quantum state} of a system $A$ is a an operator $\rho \in \Pos(A)$ such that $\tr{\rho}=1$, the set of quantum states for a system $A$ is denoted $\D(A)$. Given $\rho, \sigma \in \Pos(A)$ we write $\rho \ll \sigma$ if $\ker(\sigma) \subseteq \ker(\rho)$ where $\ker(X) \defed \{\ket{\psi} \in A \,:\, X\ket{\psi} = 0 \}$. We say a bipartite state $\rho_{XE}\in \D(XE)$ is \emph{classical-quantum} (cq) if it can be written in the form $\sum_{x} \outer{x} \otimes \rho_{E,x}$ where $\{\ket{x}\}_x$ is some predefined orthonormal (classical) basis and $\rho_{E,x} \in \Pos(E)$ such that $\sum_x \rho_{E,x} \in \D(E)$. We take the logarithm to be base 2 (denoted $\log$). For $n \in \NN$, we use the notation $[n] = \{0,1,\dots, n-1\}$.

We say a collection of operators $\{M_x\}_x$ is a POVM on system $A$ if for all $x$ we have $M_x \in \Pos(A)$ and $\sum_x M_x = \id_A$. If in addition $\{M_x\}_x$ are all projectors then we refer to this as a projective measurement or a PVM. A POVM $\{M_x\}_x$ is called $\emph{extremal}$ if for any POVMs $\{F_x\}_x$, $\{G_x\}_x$ and $0 < \lambda < 1$ such that $M_x = \lambda F_x + (1-\lambda) G_x$ for all $x$, we must have $F_x = G_x = M_x$. Note that extremal POVMs can have at most $d^2$ outcomes where $d$ is the dimension of the quantum system~\cite{DPP05}. 

In order to define the Rényi entropy families relevant to this work we must first define the sandwiched Rényi divergences, first introduced in~\cite{MDSFT, WWY13}.
\begin{definition}[Sandwiched Rényi divergences]
	Let $\alpha \in [\tfrac12,1) \cup (1,\infty)$, let $\rho, \sigma \in \Pos(A)$ with $\rho \neq 0$ and $\rho \ll \sigma$. Then the \emph{sandwiched Rényi divergence of order $\alpha$} is defined as 
	\begin{equation}
		\swD{\alpha}(\rho \|\sigma) \defed \frac{1}{\alpha-1} \log \frac{\tr{\left(\sigma^{\frac{1-\alpha}{2\alpha}} \rho \sigma^{\frac{1-\alpha}{2\alpha}}\right)^\alpha}}{\tr{\rho}} \,.
	\end{equation}
\end{definition}

We can use the above divergence to define two families of conditional entropies. 
\begin{definition}[Conditional sandwiched Rényi entropies]
	Let $\alpha \in [\tfrac12,1)\cup (1 , \infty)$ and let $\rho_{AB} \in \D(AB)$. Then we define
	\begin{equation}
		\swHu{\alpha}(A|B) := \sup_{\sigma_B \in \D(B)} -\swD{\alpha}(\rho_{AB} \| \id_A \otimes \sigma_B)
	\end{equation}
	and 
	\begin{equation}
		\swHd{\alpha}(A|B) := -\swD{\alpha}(\rho_{AB} \| \id_A \otimes \rho_B)\,.
	\end{equation}
\end{definition}

We also need a notion of a non-conditional Rényi entropy which we define as follows.
\begin{definition}[Rényi entropy]
	Let $\alpha \in (0,1) \cup (1,\infty)$ and $\rho_A \in \D(A)$. Then the Rényi entropy of system $A$ of order $\alpha$ is defined as 
	\begin{equation}
		H_\alpha(A) \defed \frac{1}{1-\alpha} \log(\tr{\rho_A^\alpha})\,.
	\end{equation}
	The cases $\alpha \in \{0,1,\infty\}$ are defined by taking the associated limits which are given by
	\begin{equation}
		\begin{aligned}
			H_0(A) \defed& \log(\mathrm{rank}(\rho_A)) \\
			H_1(A) \defed& -\tr{\rho_A \log(\rho_A)} \\
			H_{\infty}(A) \defed& -\log(\|\rho_A\|_{\infty})\,.
		\end{aligned}
	\end{equation}
	We will use the notation $H(A) \defed H_1(A)$ and $H_{\min}(A) \defed H_{\infty}(A)$. 
\end{definition}

In addition, for a bipartite state $\rho_{AB} \in \D(AB)$ we also define the \emph{conditional von Neumann entropy} as 
\begin{equation}
	H(A|B) \defed H(AB) - H(B)
\end{equation}
and the \emph{conditional min-entropy} as 
\begin{equation}
	H_{\min}(A|B) \defed -\log \inf\{\lambda \in \RR \, \mid \, \exists\, \sigma_B \in \D(B) \text{ s.t. } \rho_{AB} \succeq \id_A \otimes \sigma_B \}\,.
\end{equation}
Note that these two conditional entropies are also limits of the conditional sandwiched Rényi entropies. That is, we have
\begin{equation}
	\lim_{\alpha \to 1} \swHu{\alpha}(A|B) = H(A|B) \qquad \text{and} \qquad \lim_{\alpha \to \infty} \swHu{\alpha}(A|B) = H_{\min}(A|B)\,.
\end{equation}

\section{The maximal intrinsic $\mathbb{H}$-randomness problem (PVMs)}\label{sec:intrinsic-rand-proj}

Let us begin with a simpler problem, which is a special case of the more general problem that we solve in Section~\ref{sec:intrinsic-rand-povm}. We consider a slightly simpler\footnote{In particular, we remove the entanglement with the measurement device and we assume that the measurement $\{M_x\}_x$ is projective.} setup than that which is depicted in Figure~\ref{fig:setup}. Alice is given a source producing a state $\rho_A$ which she then measures using a projective measurement $\{M_x\}_x$ producing a classical random variable $X$. However, suppose an eavesdropper Eve holds side-information about the source. This side information is modeled by a bipartite quantum state $\rho_{AE}$, which has the constraint that $\ptr{E}{\rho_{AE}} = \rho_A$. Note that the dimension of $E$ is a priori unbounded.

After the measurement is performed by Alice, the information is represented by the cq-state 
\begin{equation}
	\rho_{XE} = \sum_x \outer{x} \otimes \rho_{E,x} 
\end{equation}
where $\rho_{E,x} = \ptr{A}{(M_x \otimes \id) \rho_{AE}}$. Using this state we can quantify the information Eve has via a conditional entropy $\mathbb{H}(X|E)$. For the moment, we do not explicitly define $\mathbb{H}$ but we note that the particular cases where $\mathbb{H}$ is the von Neumann entropy and the min-entropy have already been studied in~\cite{MCSWFSA23}.

In order to conclude that the outcome of the measurement $\{M_x\}_x$ on the state $\rho_A$ is \emph{intrinsically random}, we must allow for any possible side information held by Eve. That is we look to solve the optimization problem 
\begin{equation}\label{eq:minE-proj}
	\begin{aligned}
		\inf & \quad \mathbb{H}(X|E) \\
		\mathrm{s.t.}& \quad \rho_{XE} = \sum_x \outer{x} \ox \ptr{A}{(M_x \otimes \id) \rho_{AE}}, \\
		& \quad \ptr{E}{\rho_{AE}} = \rho_A, \qquad \rho_{AE} \in \D(AE)\,,
	\end{aligned}
\end{equation}
where we minimize over all possible $E$ systems that could be held by the adversary. The following lemma (see also~\cite{HB24}) shows that it is sufficient to consider any purification $\rho_{AE}$ of $\rho_A$.
\begin{lemma}\label{lem:purification-proj}
	Let $\HH$ be a conditional entropy satisfying:
	\begin{enumerate}
		\item (Strong subadditivity) $\HH(A|BC) \leq \HH(A|B)$
		\item (Isometric invariance) For any isometry $V: B \rightarrow B'$ we have $\HH(A|B)_{\rho_{AB}} = \HH(A|B')_{V \rho_{AB} V^\dagger}$
	\end{enumerate}
	then the optimal solution to the optimization problem~\eqref{eq:minE-proj} is obtained by any purification $\rho_{AE}$ of $\rho_A$.  
	\begin{proof}
		Take any feasible point of~\eqref{eq:minE-proj} such that $\rho_{AE}$ is a mixed state. Then if $\rho_{AEE'}$ is a purification of $\rho_{AE}$ we have $\HH(X|EE') \leq \HH(X|E)$ by strong subadditivity. Thus we can always restrict to pure $\rho_{AE}$ when minimizing.
		
		Any two purifications $\rho_{AE}$, $\rho_{AE'}$ of $\rho_A$ can be related by an isometry on the purifying system. I.e., suppose without loss of generality that $\dim(E) \leq \dim(E')$. Then there exists an isometry $V:E\to E'$ such that $\rho_{AE'} = (\id_A \otimes V) \rho_{AE} (\id_A \otimes V^{\dagger})$. As the measurement occurs only on the $A$ system we then also have $(\id_X \otimes V) \rho_{XE} (\id_X \otimes V^{\dagger}) = \rho_{XE'}$. By the isometric invariance property of $\HH$ we therefore have $\HH(X|E) = \HH(X|E')$. Thus all purifications result in the same optimal value. 
	\end{proof}
\end{lemma}

Lemma~\ref{lem:purification-proj} enables us to quantify the intrinsic randomness of a measurement $\{M_x\}_x$ on a state $\rho_A$ as quantified by $\HH$ by evaluating $\HH(X|E)$ where $E$ is any purification of $\rho_A$.\footnote{This also implies that without loss of generality we can assume that $\dim(A) = \dim(E)$.} In this work we are interested in what we call the \emph{maximal intrinsic $\HH$-randomness} of the state $\rho_A$, as such we allow ourselves to maximize $\HH(X|E)$ over all projective measurements. This leads to the following definition. 

\begin{definition}[Maximal intrinsic $\HH$-randomness (projective)]
	\label{def:max_intrinsinc_randomness-proj}
	Let $\HH$ be a conditional entropy, let $\rho_A$ be a quantum state and let $\rho_{AE}$ be any purification of $\rho_A$. Then we define the \emph{maximal intrinsic $\HH$-randomness for projective measurements} of $\rho_A$ as 
	\begin{equation}
		\begin{aligned}
			\maxintrinsicpvm{\HH}{\rho_A} \defed \sup_{\{M_x\}_x} &\quad \HH(X|E) \\
			\mathrm{s.t.}& \quad \rho_{XE} = \sum_x \outer{x} \ox \ptr{A}{(M_x \otimes \id) \rho_{AE}}, \\
			& \quad \{M_x\}_x \text{ is a projective measurement}.
		\end{aligned}
	\end{equation}
\end{definition} 

In~\cite{MCSWFSA23}, closed form expressions for the von Neumann entropy $H$ and the min-entropy $H_{\min}$ were derived. In particular, the authors found that 
\begin{equation}\label{eq:int-rand-pvm-H}
	\maxintrinsicpvm{H}{\rho_A} = \log(d_A) - H(A)
\end{equation}
and 
\begin{equation}\label{eq:int-rand-pvm-Hmin}
		\maxintrinsicpvm{H_{\text{min}}}{\rho_A} = \log(d_A) - 2\log \tr{\sqrt{\rho_A}}.
\end{equation}

In the following subsection we will generalize these results to the sandwiched Rényi entropies $\swHu{\alpha}$ and $\swHd{\alpha}$. Using this we can recover the results from~\cite{MCSWFSA23} by taking appropriate limits. 

\subsection{Maximal intrinsic Rényi-randomness (PVMs)}

To begin we have the following result that tells us we can assume, without loss of generality, that our optimal projective measurement is a rank-one projective measurement. 
\begin{lemma}[Restriction to rank-one measurements]\label{lem:rank1}
	Let $\alpha \in [1/2, 1)\cup(1,\infty)$, then for $\HH$ either $\swHd{\alpha}$ or $\swHu{\alpha}$ we have
	\begin{equation}
		\begin{aligned}
			\maxintrinsicpvm{\HH}{\rho_A} = \sup_{\{M_x\}_x} &\quad \HH(X|E) \\
			\mathrm{s.t.}& \quad \rho_{XE} = \sum_x \outer{x} \ox \ptr{A}{(M_x \otimes \id) \rho_{AE}}, \\
			& \quad \{M_x\}_x \text{ is a rank-one projective measurement}.
		\end{aligned}
	\end{equation}
	\begin{proof}
		Consider any PVM $\{N_x\}_x$ (not necessarily rank-one) and let $r = \max\limits_x \mathrm{rank}(N_x)$. Then there exists a collection of projectors $\{M_{x,y}\}_{x,y}$ such that $\sum_{y=1}^r M_{x,y} = N_x$ such that each $M_{x,y}$ is either a rank-one projector or 0. If we define the two classical-quantum states 
		\begin{equation}
			\begin{aligned}
				\rho_{XE} &= \sum_x \outer{x} \otimes \ptr{A}{(N_x \otimes \id_E) \rho_{AE}} \\
				\rho_{XYE} &= \sum_{xy} \outer{xy} \otimes \ptr{A}{(M_{x,y} \otimes \id_E) \rho_{AE}}
			\end{aligned}
		\end{equation}
		then $\ptr{Y}{\rho_{XYE}} = \rho_{XE}$. By Lemma~\ref{lem:classical-info-ent} we know that for all $\alpha \in [1/2,1)\cup(1,\infty)$ we have
		\begin{equation}
			\swHd{\alpha}(XY|E) \geq \swHd{\alpha}(X|E) \quad \text{and} \quad \swHu{\alpha}(XY|E) \geq \swHu{\alpha}(X|E)\,. 
		\end{equation}
		Overall this shows that for any non rank-one measurement we can find a rank-one measurement whose conditional entropy is at least as large and hence the result follows.
	\end{proof}
\end{lemma}

Using the above lemma we can restrict the intrinsic randomness to an optimization over rank-one projective measurements. This gives us sufficient structure to prove the following closed form expressions of the intrinsic sandwiched Rényi-randomness, the proof is given in Appendix~\ref{app:sw}.
\begin{theorem}\label{thm:int-rand-proj-sw}
	Let $\rho_A \in \D(A)$ and let $d_A$ be the dimension of system $A$. Then we have
	\begin{equation}
		\maxintrinsicpvm{\swHu{\alpha}}{\rho_A} = \log(d_A) - H_{\frac{\alpha}{2\alpha-1}}(A)
	\end{equation}
	for all $\alpha > 1$ and we have
	\begin{equation}
		\begin{aligned}
			\maxintrinsicpvm{\swHd{\alpha}}{\rho_A} = \log(d_A) - H_{\frac{1}{\alpha}}(A)		\,.
		\end{aligned}
	\end{equation}
	for all $\alpha \in [1/2,1) \cup (1,\infty)$.
\end{theorem}

We note that the proof of Theorem~\ref{thm:int-rand-proj-sw} can also be extended to the limits $\alpha \to 1$ and $\alpha \to \infty$. In particular, by taking these limits we recover the results
\begin{equation}
	\maxintrinsicpvm{H}{\rho_A} = \log(d_A) - H(A) \quad \text{and} \quad \maxintrinsicpvm{H_{\text{min}}}{\rho_A} = \log(d_A) - 2\log\tr{\sqrt{\rho_A}} 
\end{equation}
which were first derived in~\cite{MCSWFSA23}. Thus our results for $\maxintrinsicpvm{\swHu{\alpha}}{\rho_A}$ nicely interpolate between the results of~\cite{MCSWFSA23}, can be derived using a single proof technique and give a clear interpretation for the form of $2 \log \tr{\sqrt{\rho_A}} = H_{1/2}(A)$. 
\begin{example}
	Consider the qubit state $\rho_{A} = \tfrac12 \outer{0} + \tfrac12 \outer{+}$, whose spectrum is equal to $\{\cos^2(\pi/8), \sin^2(\pi/8)\}$. In Figure~\ref{fig:ir-proj-example} we plot $\maxintrinsicpvm{\swHd{\alpha}}{\rho_A}$ and $\maxintrinsicpvm{\swHu{\alpha}}{\rho_A}$ for a range of $\alpha$. In particular we can see that $\maxintrinsicpvm{\swHd{\alpha}}{\rho_A} \leq \maxintrinsicpvm{\swHu{\alpha}}{\rho_A}$, with an equality in the limit $\alpha \to 1$. However, for other values of $\alpha$ they diverge with
	\begin{equation}
		\lim_{\alpha \to \infty} \maxintrinsicpvm{\swHu{\alpha}}{\rho_A} = \maxintrinsicpvm{\swHd{2}}{\rho_A}
	\end{equation}
	and
	\begin{equation}
		\lim_{\alpha \to \infty} \maxintrinsicpvm{\swHd{\alpha}}{\rho_A} = \log(d_A) - \log \mathrm{rank}(\rho_A)\,,
	\end{equation}
	which evaluates to $0$ in this example as $\mathrm{rank}(\rho_A) = d_A$.
\end{example}

\begin{figure}
	\definecolor{mycolor2}{rgb}{0.00000,0.44700,0.74100}%
	\definecolor{mycolor4}{rgb}{0.85000,0.32500,0.09800}%
	\definecolor{mycolor3}{rgb}{0.92900,0.69400,0.12500}%
	\definecolor{mycolor1}{rgb}{0.49400,0.18400,0.55600}%
	\centering
	\begin{tikzpicture}[scale = 0.9]	
		\begin{axis}[%
			width=4in,
			height=2.8in,
			scale only axis,
			xmin=0.5,
			xmax=4,
			ymin=0,
			ymax=1,
			grid=major,
			ytick = {0.0, 0.1, 0.2, 0.3, 0.4, 0.5, 0.6, 0.7,0.8,0.9,1.0},
			xtick = {0.5, 1, 1.5,2,2.5,3,3.5,4},
			axis background/.style={fill=white},
			xlabel={$\alpha$},
			ylabel={Maximal intrinsic $\HH$-randomness (bits)},
			legend style={at={(0.84,0.9)},legend cell align=left, align=left, draw=white!15!black,row sep=3pt}
			]
			\addplot[color=mycolor1, line width=1.5pt] table[col sep=comma] {./data/swHu-ex.csv};
			\addlegendentry{$\maxintrinsicpvm{\swHu{\alpha}}{\rho_A}$};
			\addplot[color=mycolor2, line width=1.5pt] table[col sep=comma] {./data/swHd-ex.csv};
			\addlegendentry{$\maxintrinsicpvm{\swHd{\alpha}}{\rho_A}$};
			\addplot[color=black, dashed, line width=1.5pt] table[col sep=comma] {./data/swHuAsym-ex.csv};
			\addlegendentry{$\lim\limits_{\alpha \to \infty} \maxintrinsicpvm{\swHu{\alpha}}{\rho_A}$};
		\end{axis}
		\node[] at (7.3,3.8) {$\rho_A = \frac14 \begin{pmatrix}
				3 & 1 \\ 1 & 1
			\end{pmatrix}$};
	\end{tikzpicture}%
	\caption{A plot of $\cR_{\swHd{\alpha}}(\rho_A)$ and $\cR_{\swHd{\alpha}}(\rho_A)$ for the state $\rho_A = \frac14 \begin{pmatrix} 3 & 1 \\ 1 & 1 \end{pmatrix}$.  }
	\label{fig:ir-proj-example}
\end{figure}

\section{The maximal $\mathbb{H}$-randomness problem}\label{sec:intrinsic-rand-povm}

We now want to generalize the setup considered in the previous section to the case where we allow Alice to measure any POVM $\{M_x\}_x$. However, this generalization is not as immediate as just replacing the maximization over projective measurement to over POVMs. To see that there are issues consider the trivial measurement $\{\id_A / n, \id_A / n, \dots , \id_A/n \}$, for any state $\rho_{AE}$ we will have that $H(X|E) = \log(n)$, in particular we can make this as large as possible by making $n$ large despite the measurement outcome carrying no information about the underlying state.

We use a framework introduced in~\cite{FRT13} where the adversary may share entanglement with the state used in the Naimark dilation of the POVM, see also~\cite{SSA23,DCZM23}. In particular, we now not only allow Eve to be entangled with the source, but we now allow for arbitrary entanglement between the source, the measurement device and Eve (see also Figure~\ref{fig:setup} for a graphical depiction). That is, we now introduce a quantum system $B$ that is held by the measurement device and can be used to help implement the POVM $\{M_x\}_x$. Thus, the shared state before the measurement is some $\rho_{ABE} \in \D(ABE)$ with the marginal $\rho_A$ still fixed. When the measurement device receives the system $A$, the most general thing it can do is apply a quantum-to-classical channel $\cE: L(AB) \to L(X)$ to produce the measurement outcomes. Such a channel can be shown to take the form 
\begin{equation}
	\cE(\rho_{AB}) = \sum_x \tr{\rho_{AB} N_x} \outer{x} 
\end{equation}
where $N_x$ is some POVM on $AB$. Importantly, however, we trust that the measurement device is actually implementing a fixed POVM $\{M_x\}_x$. Thus, we must have that for all $\sigma_A \in \D(A)$, 
\begin{equation}
	\tr{(\sigma_A \otimes \rho_B) N_x} = \tr{\sigma_A M_x}
\end{equation}
or equivalently, $\ptr{B}{N_x(\id_A \otimes \rho_B)} = M_x$. By considering a Naimark dilation this is equivalent to there existing a projective measurement $\{P_x\}_x$ that acts on the joint system $AB$\footnote{The dimension of $B$ may change and hence $\rho_B$ also but as we don't specify its dimension there are no issues here.} which satisfies what we refer to as the \emph{consistency condition}~\cite{DCZM23}: 
\begin{equation}
	\ptr{B}{P_x(\id_A \otimes \rho_B)} = M_x\,.
\end{equation}

\begin{figure}[t]
	\begin{center}
		\begin{tikzpicture}
			\def\x{1.5}
			\draw[thick, rounded corners = 1, fill = gray] (0,0,\x) rectangle (\x,\x,\x);
			\draw[thick, rounded corners = 1, fill = gray] (0,\x,\x) -- (0,\x,0) -- (\x,\x,0) -- (\x,\x,\x) -- cycle;
			\draw[thick, rounded corners = 1, fill = gray] (\x,\x,0) -- (\x,0,0) -- (\x,0,\x) -- (\x,\x,\x) -- cycle;
			\node[] at (0.55,0.1,1) {\small $\{M_x\}_x$};
			
			\draw[thick] (-.4,0.15) -- (.7,0.15);
			\draw[thick] (.7, 0.15) arc (0:180:.56);
			\draw[thick,->] (.15,.15) -- (.55,.5);
			
			\draw[thick, ->] (1.25,.5) -- (2.25,.5);
			\node[] (X) at (2.5,.5) {\fontsize{18}{12}$X$};
			\node[] (Z) at (6.5,.5) {\fontsize{18}{12}$Z$};
			\node[] (E) at (2.5,-1.25) {\fontsize{18}{12}$E$};
			
			\draw[thick, ->] (X) -- (Z);
			\node[] at (4.45, .72) {Randomness};
			\node[] at (4.45, .32) {extraction};
			
			\def\x{1.5}
			\def\ex{0};
			\def\ey{-2.5};
			
			\node[left] (A) at (-3.0, -.8) {\fontsize{14}{12} $\rho_{ABE}$};
			\draw[thick, dotted, ->] (A) ..controls (-2.5,0.5) .. (-.6,0.5);
			\draw[thick, dotted, ->] (A) ..controls (-2,0.) .. (-.6,0.);
			\draw[thick, dotted, ->] (A) ..controls (-2.5,-1.3) .. (E);

			\node[right] (C) at (3.2,.-.5) {\fontsize{14}{12} $\rho_{ZE} \approx_\epsilon \tau_k \otimes \rho_{E}$};
		\end{tikzpicture}
	\end{center}
	\caption{\textbf{Schematic of maximal extractable randomness.} We consider a setup where a trusted state $\rho_A \in \D(A)$ is measured by some trusted measurement $\{M_x\}_x$ producing a random variable $X$. However, we allow for the state to also be entangled with a system $B$ which is held by the measurement device and a system $E$ that is held by the adversary. Thus the joint state $\rho_{ABE}$ is arbitrary apart from the condition that the marginal on $A$ is fixed to $\rho_A$. Then, after the cq-state $\rho_{XE}$ is produced, a randomness extraction protocol is applied to $X$ to produce a new state $\rho_{ZE}$ where $Z$ is an $\epsilon$-secure string of $k$ uniformly random bits. If we fix $\rho_A$, how large can we make $k$?}
	\label{fig:setup}
\end{figure}

Note that such Naimark dilations are not unique and hence we further allow Eve to choose $\{P_x\}_x$ and $\rho_{B}$ (equivalently $\rho_{ABE}$) in order to minimize her uncertainty. Therefore, this new setting begins with a state $\rho_{ABE}$ where the system $E$ is still held by Eve. Then the state after Alice's measurement device implements the POVM $\{M_x\}_x$ is given by 
\begin{equation}
	\rho_{XE} = \sum_x \outer{x} \otimes \ptr{AB}{(P_x \otimes \id_E) \rho_{ABE}}\,.
\end{equation} 
In analogy to~\eqref{eq:minE-proj}, for a fixed state $\rho_A$, a fixed POVM $\{M_x\}$ and a conditional entropy $\HH$, we can define the intrinsic $\HH$-randomness by the optimization
\begin{equation}\label{eq:minE-povm}
	\begin{aligned}
		\inf & \quad \mathbb{H}(X|E) \\
		\mathrm{s.t.}& \quad \rho_{XE} =\sum_x \outer{x} \otimes \ptr{AB}{(P_x \otimes \id_E) \rho_{ABE}}, \\
		& \quad \ptr{BE}{\rho_{ABE}} = \rho_A, \quad \rho_{ABE} \in \D(AE), \\
		& \quad \ptr{B}{P_x(\id_A \otimes \rho_B)} = M_x\,,  \\
		& \quad \{P_x\}_x \text{ is a projective measurement}.
	\end{aligned}
\end{equation}
This optimization is evidently more complex that~\eqref{eq:minE-proj} as the optimization is now not only over all extensions of $\rho_A$, but also over all Naimark dilations of $\{M_x\}_x$. However, importantly, using the same proof technique as Lemma~\ref{lem:purification-proj} we can show that without loss of generality we can take $\rho_{ABE}$ to be a purification of $\rho_{AB}$. Thus, the remaining optimization is just over all Naimark dilations $(\{P_x\}_x, \rho_B)$ and we can take $\dim(E) = \dim(AB)$. Note that any valid Naimark dilation $(\{P_x\}_x, \rho_B)$ of $\{M_x\}_x$ gives a feasible point as we can always take $\rho_{ABE}$ to be a purification of $\rho_A \otimes \rho_B$. However, the optimization is over all possible $\rho_{ABE}$ and hence there can in principle be entanglement between $A$ and $B$.

Using the observation regarding purifications, we arrive at our definition of \emph{intrinsic randomness} for a pair $(\rho_A, \{M_x\}_x)$.
\begin{definition}\label{def:int-rand-povm}
	Let $\rho_A \in \D(A)$ and let $\{M_x\}_x$ be a POVM on the system $A$. Then the \emph{intrinsic $\HH$-randomness of the pair $(\rho_A, \{M_x\}_x)$} is defined as 
	\begin{equation}
			\begin{aligned}
					\intrinsicpovm{\HH}{\rho_A}{\{M_x\}_x} :=	\inf & \quad \mathbb{H}(X|E) \\
	\mathrm{s.t.}& \quad \rho_{XE} =\sum_x \outer{x} \otimes \ptr{AB}{(P_x \otimes \id_E) \outer{\rho}_{ABE}}, \\
& \quad \ptr{B}{\rho_{AB}} = \rho_A, \quad \rho_{AB} \in \D(AB), \\
& \quad \ptr{B}{P_x(\id_A \otimes \rho_B)} = M_x\,,  \\
& \quad \{P_x\}_x \text{ is a projective measurement}.
			\end{aligned}
	\end{equation}
	where $\ket{\rho}_{ABE} = \sum_{k} \sqrt{\rho_{AB}} \ket{k}_{AB}\otimes\ket{k}_{E}$ is a purification of $\rho_{AB}$. 
\end{definition}

This simplifies the optimization somewhat by reducing it to an optimization just over Naimark dilations and bipartite states $\rho_{AB}$. However, the optimization over all Naimark dilations is a priori difficult due to the possibly unbounded dimension of $B$. Our goal will be to tame this optimization for a special class of POVMs, namely extremal rank-one POVMs. Before that, we show that the adversarial model takes care of the trivial measurement problem highlighted at the beginning of the section.

\begin{example}[Trivial POVMs produce no intrinsic randomness]\label{ex:trivial}
	Consider the trivial POVM $\{M_x\}_x$ where $M_x = p(x) \id_A$ with $p(x)$ being some probability distribution. Such a POVM can be implemented with the Naimark dilation
	\begin{equation}
		P_x = \id_A \otimes \outer{x}  \qquad \text{and} \qquad \rho_B = \sum_x p(x) \outer{x}\,.
	\end{equation}
	Define $\ket{\psi}_{BE} = \sum_x \sqrt{p(x)} \ket{xx}$ to be a purification of $\rho_B$ and let the tripartite state be $\rho_{A} \otimes \outer{\psi}_{BE}$ then the post-measurement state takes the form
	\begin{equation}
		\rho_{XE} = \sum_x p(x) \outer{x}_X \otimes \outer{x}_E
	\end{equation}
	which gives $\HH(X|E) = 0$. Note in this example we did not require Eve to share entanglement with $\rho_A$, although this is natural as the cq-state $\rho_{XE}$ is independent of $\rho_A$. 
\end{example}

Using Definition~\ref{def:int-rand-povm} we can define the maximal intrinsic randomness of the state $\rho_A$ for POVMs (analogous to Definition~\ref{def:max_intrinsinc_randomness-proj}) by maximizing over all POVMs.

\begin{definition}[Maximal intrinsic $\HH$-randomness]
	Let $\rho_A \in \D(A)$ and let $\HH$ be a conditional entropy. Then we define the \emph{maximal intrinsic $\HH$-randomness of $\rho_A$} as
	\begin{equation}
		\maxintrinsicpovm{\HH}{\rho_A} = \sup_{\{M_x\}_x} \intrinsicpovm{\HH}{\rho_A}{\{M_x\}_x}\,.
	\end{equation}
\end{definition}

In a similar fashion to Section~\ref{sec:intrinsic-rand-proj} we now look to find closed form solutions for $\maxintrinsicpovm{\HH}{\rho_A}$ for the special cases where $\HH$ corresponds to a sandwiched Rényi entropy. 

\subsection{Maximal Intrinsic Rényi randomness}

In order to find closed form expressions for $\maxintrinsicpovm{\HH}{\rho_A}$, we need a way of solving the minimization over the Naimark dilations implicit in the definition. Whilst it is not obvious how to solve such a problem in general, it would be sufficient if we could just solve the problem for the subset of measurements $\{M_x\}_x$ that maximize $\intrinsicpovm{\HH}{\rho_A}{\{M_x\}_x}$. The following lemma shows that we can consider rank-one extremal POVMs without loss of generality. 

\begin{lemma}[Reduction to extremal rank-one POVMs]
	\label{lem:extremal_povm_is_optimal}
	Let $\alpha \in [\tfrac12,1) \cup (1,\infty)$, let \(\HH\) be either \(\swHd{\alpha}\) or \(\swHu{\alpha}\) and let be $\rho_A \in \D(A)$. Then
	\begin{equation}
		\maxintrinsicpovm{\HH}{\rho_A} = \sup_{\substack{\mathrm{Extremal~rank-one}\\ \{M_x\}_x}} \intrinsicpovm{\HH}{\rho_A}{\{M_x\}_x}\,.
	\end{equation}
	That is, we can restrict our optimization to extremal rank-one POVMs without loss of generality.
	\begin{proof}
		We begin by proving that $\{M_x\}_x$ can be assumed to be extremal. Suppose that $\{M_x\}_x$ is not extremal, then there exists a convex decomposition
		\begin{equation}
			M_x = \sum_k \lambda_k N_{x,k}
		\end{equation}
		where for each $k$, $\{N_{x,k}\}_x$ is an extremal POVM and $\{\lambda_k\}_k$ form a probability distribution. By Lemma~\ref{lem:convexity} we know that the map  $ 2^{\intrinsicpovm{\HH}{\rho_A}{\{M_x\}_x}}$ is convex on the set of POVMs and hence we have
		\begin{equation}
			2^{\intrinsicpovm{\HH}{\rho_A}{\{M_x\}_x}} \leq \sum_k \lambda_k\,2^{\intrinsicpovm{\HH}{\rho_A}{\{N_{x,k}\}_x}}
			\leq \max_k\,\, 2^{\intrinsicpovm{\HH}{\rho_A}{\{N_{x,k}\}_x}}\,.
		\end{equation}
		Since $\log$ is a monotonically increasing function we therefore have that
		\begin{equation}
			\intrinsicpovm{\HH}{\rho_A}{\{M_x\}_x} \leq \max_k\,\, \intrinsicpovm{\HH}{\rho_A}{\{N_{x,k}\}_x}.
		\end{equation}
		Thus we can always find an extremal POVM that obtains an intrinsic randomness at least as large as any non-extremal POVM and hence we can restrict the optimization to extremal POVMs.
		
		Now suppose $\{M_x\}_x$ is extremal but not rank-one. Consider the spectral decomposition $M_x = \sum_y \lambda_{x,y} \outer{v_{x,y}}$ and define a new rank-one POVM $N_{x,y} = \lambda_{x, y} \outer{v_{x,y}}$ where the $(x,y)$ range only over indexes where $\lambda_{x,y} > 0$. By Lemma~\ref{lem:extremal-to-rankone} this is also an extremal POVM. Now consider any dilation $(\{P_{x,y}\}_{x,y}, \sigma_B)$ for $\{N_{x,y}\}_{x,y}$, if we define $Q_x = \sum_{y} P_{x,y}$ then 
		\begin{equation}
			\ptr{B}{Q_x(\id_A \otimes \sigma_B)} = \sum_y \ptr{B}{P_{x,y}(\id_A \otimes \sigma_B)} = \sum_y N_{x,y} = M_x
		\end{equation}
		thus $(\{Q_x\}_x, \sigma_B)$ is a dilation for $M_x$. But then we have
		\begin{equation}
			\inf_{(\{P_{x,y}\}_{x,y}, \sigma_B)} \HH(XY|E) \geq \inf_{(\{Q_x\}_x, \sigma_B)}\HH(X|E) 
		\end{equation}
		which follows from Lemma~\ref{lem:classical-info-ent}. Overall we have shown that given any POVM we can always find a rank-one extremal POVM that has an intrinsic randomness at least as large. 
 	\end{proof}
\end{lemma}

Thus we have reduced the optimization to extremal rank-one POVMs. It turns out, as the following lemma implies, that for extremal POVMs the infimum over Naimark dilations is trivial, in the sense that we can without loss of generality assume that $B$ is entangled with neither $A$ nor $E$. The proof of this lemma is deferred to Appendix~\ref{app:max-sw-povm} (see Lemma~\ref{lem:separable-probe-state-app}).

\begin{lemma}
	\label{lem:separable-probe-state}
	Let $\left\{M_x\right\}_x$ be an extremal rank-one POVM on a system $A$, let $\left(\left\{P_x\right\}_x, \rho_B\right)$ be any dilation of $\{M_x\}_x$ and let $\rho_{ABE} \in \D(ABE)$ be any pure state with marginals $\rho_A$ and $\rho_B$. Then the post-measurement state
	\begin{equation}
		\rho_{XE} = \sum_x \outer{x} \otimes \ptr{AB}{(P_x \otimes \id_E) \rho_{ABE}}
	\end{equation}
	can be written as 
	\begin{equation}
		\rho_{XE} = \sum_x \outer{x} \otimes \ptr{A}{(M_x \otimes \id_E)\rho_{AE}}\,.
	\end{equation}
\end{lemma}

Note that due to the a priori potential for entanglement between $ABE$ the above lemma is not trivial. For instance, if we take Example~\ref{ex:trivial} then the considered POVM is not extremal and the above statement does not hold. Explicitly, we find that
\begin{equation}
	\ptr{AB}{(P_x\otimes \id_E) \rho_{ABE}} = \outer{x}_E \qquad \neq \qquad \ptr{A}{(M_x \otimes \id_E)} = \tfrac1n \rho_E\,.
\end{equation}
Overall, the above lemma implies that for extremal rank-one POVMs we can completely ignore the optimization over Naimark dilations. Combined with Lemma~\ref{lem:extremal_povm_is_optimal}, this reduces the problem of computing $\maxintrinsicpovm{\HH}{\rho_A}$ significantly. In effect, we have shown that 
\begin{equation}
\begin{aligned}
		\maxintrinsicpovm{\HH}{\rho_A} = \sup& \quad \HH(X|E) \\
				\mathrm{s.t.}& \quad \rho_{XE} = \sum_x \outer{x} \otimes \ptr{A}{(M_x \otimes \id) \outer{\rho}_{AE}} \\
				& \quad \{M_x\}_x \text{ is an extremal rank-one POVM}
\end{aligned}
\end{equation}
where $\ket{\rho}_{AE}$ is any purification of $\rho_A$.
This is effectively the same problem as was solved in Section~\ref{sec:intrinsic-rand-proj} but we optimize over all extremal rank-one POVMs rather than just rank-one PVMs. This gives us enough structure to finally solve the optimization analytically. This results in the following theorem, the proof of which is deferred to Appendix~\ref{app:max-sw-povm}.
\begin{theorem}\label{thm:int-rand-povm-sw}
	Let $\rho_A \in \D(A)$ and let $d_A$ be the dimension of system $A$. Then we have, for all $\alpha > 1$,
	\begin{equation}
		\maxintrinsicpovm{\swHu{\alpha}}{\rho_A} = 2\log(d_A) - H_{\frac{\alpha}{2\alpha-1}}(A)\,.
	\end{equation}
	Moreover, this quantity is achieved by an extremal rank-one POVM $\{M_x\}_x$ if and only if
	\begin{equation}
		\frac{\tr{\rho_A^{\frac{\alpha}{2\alpha-1}}M_x}}{\tr{\rho_A^{\frac{\alpha}{2\alpha-1}}}}= \frac{1}{d_A^2}\,
	\end{equation}
	for all $x \in [d_A^2]$. Additionally, for all $\alpha \in [1/2,1) \cup (1,\infty)$, we have
	\begin{equation}
		\begin{aligned}
			\maxintrinsicpovm{\swHd{\alpha}}{\rho_A} = 2\log(d_A) - H_{\frac{1}{\alpha}}(A)		\,.
		\end{aligned}
	\end{equation}
	Moreover, this quantity is achieved by an extremal rank-one POVM $\{M_x\}_x$ if and only if
	\begin{equation}
		\frac{\tr{\rho_A^{1/\alpha}M_x}}{\tr{\rho_A^{1/\alpha}}} = \frac{1}{d_A^2}\,
	\end{equation}
	for all $x \in [d_A^2]$.
\end{theorem}

It is worth making a few remarks about this result. Firstly, we have that 
\begin{equation}
	\maxintrinsicpovm{\HH}{\rho_A} = \log(d_A) + \maxintrinsicpvm{\HH}{\rho_A}\,.
\end{equation}
This increase in the maximal amount of intrinsic randomness is due to the fact that extremal POVMs can have up to $d_A^2$ outcomes whereas projective measurements have at most $d_A$ outcomes. Furthermore, whilst we did not consider Eve to have side information on the measurement device when deriving the result for $\maxintrinsicpvm{\HH}{\rho_A}$, our construction shows that even if we were to include this side information in the analysis of Section~\ref{sec:intrinsic-rand-proj} then this would not change the value of $\maxintrinsicpvm{\HH}{\rho_A}$. This is because rank-one projective measurements are extremal and hence the above reductions apply. This shows that the results of~\cite{MCSWFSA23} also hold in the presence of side-information with the measurement device.

Another interesting point is that whilst $\maxintrinsicpvm{\HH}{\rho_A}$ can be equal to 0 (e.g. for the maximally mixed state) we instead have that $\maxintrinsicpovm{\HH}{\rho_A} \geq \log(d_A)$. Thus for all quantum states, it is possible to extract a significant amount of randomness from them. However, it should be noted that achieving this theoretical maximum may be extremely challenging in practice this is because, whilst the maximal intrinsic randomness  $\maxintrinsicpovm{\HH}{\rho_A}$ is continuous in $\rho_A$, the intrinsic randomness $\intrinsicpovm{\HH}{\rho_A}{\{M_x\}_x}$ is not continuous on the set of POVMs for a fixed state $\rho_A$. 

The possibility for discontinuities is evident from the proof of Theorem~\ref{thm:int-rand-povm-sw} (see Theorem~\ref{thm:maxint-app} for the proof). In the proof we establish the existence of a sequence of extremal rank-one POVMs that achieve an intrinsic randomness rate approaching $\maxintrinsicpovm{\HH}{\rho_A}$. However, the limiting POVM for this sequence is not extremal and actually has an intrinsic randomness rate far lower than $\maxintrinsicpovm{\HH}{\rho_A}$. Thus our proof says nothing about the achievability of $\maxintrinsicpovm{\HH}{\rho_A}$. Nevertheless we can create examples in which it is achievable, for instance consider a qubit case. 

\begin{example}[Optimal POVMs for qubits]\label{ex:opt-qubit}
Let us define a parameterized POVM $\{\outer{\psi_x(t)}\}_x$ with $t \in [1/2, 1]$ and
\begin{equation}
	\begin{aligned}
		\ket{\psi_0(t)} &= \frac{1}{2\sqrt{t}}\,\ket{0}\\
		\ket{\psi_1(t)} &= \sqrt{\frac{4 t - 1}{12 t}}\,\ket{0}+\sqrt{\frac13}\,\ket{1}\\
		\ket{\psi_2(t)} &= \sqrt{\frac{4 t - 1}{12 t}}\,\ket{0}+\sqrt{\frac13}\mathrm{e}^{\frac{2\mathrm{i}\pi}{3}}\,\ket{1}\\
		\ket{\psi_3(t)} &= \sqrt{\frac{4 t - 1}{12 t}}\,\ket{0}+\sqrt{\frac13}\mathrm{e}^{\frac{4\mathrm{i}\pi}{3}}\,\ket{1}\,.
	\end{aligned}
\end{equation}
One can show that this is an extremal POVM for all $t \in [1/2, 1]$. Now consider the qubit state $\rho_A = \lambda \outer{0} + (1-\lambda)\outer{1}$ where without loss of generality we assume that $\rho_A$ is diagonal in the computational basis and furthermore that $\lambda \geq 1/2$. Then a short calculation shows that for $\alpha > 1$ we can choose $t = \frac{\lambda^{\tfrac{\alpha}{2\alpha-1}}}{\lambda^{\tfrac{\alpha}{2\alpha-1}} + (1-\lambda)^{\tfrac{\alpha}{2\alpha-1}}} \in [1/2, 1]$ and this measurement achieves 
\begin{equation}
	\intrinsicpovm{\swHu{\alpha}}{\rho_A}{\{M_x(t)\}} = \maxintrinsicpovm{\swHu{\alpha}}{\rho_A}\,.
\end{equation}
I.e., there exists a choice of POVM that achieves the maximal intrinsic randomness as measured by $\swHu{\alpha}$. A similar calculation shows that for all $\alpha \in [1/2,1) \cup (1, \infty)$, if we choose $t = \lambda^{1/\alpha}/(\lambda^{1/\alpha} + (1-\lambda)^{1/\alpha}) \in [1/2, 1]$ then 
\begin{equation}
	\intrinsicpovm{\swHd{\alpha}}{\rho_A}{\{M_x(t)\}} = \maxintrinsicpovm{\swHd{\alpha}}{\rho_A}\,.
\end{equation}
and hence the maximal intrinsic randomness as measured with respect to $\swHd{\alpha}$ can be achieved by an explicit qubit POVM. Overall this shows that for a qubit system, the maximal intrinsic randomness quantities can always be achieved.
\end{example}

Discontinuities in $\intrinsicpovm{\HH}{\rho_A}{\{M_x\}_x}$ implies an interesting physical problem. In particular, the POVMs used in the proof of Theorem~\ref{thm:int-rand-povm-sw} converge to a discontinuity and hence to get high intrinsic randomness rates with these POVMs may require extreme precision in order to not encounter the discontinuity. It would be interesting to understand two-fold:
\begin{enumerate}
	\item Are there POVMs for all states $\rho_A$ that achieve $\maxintrinsicpovm{\HH}{\rho_A}$? 
	\item Are there POVMs that achieve a high intrinsic randomness rate $\intrinsicpovm{\HH}{\rho_A}{\{M_x\}_x}$ and are far from any discontinuities?
\end{enumerate}
In particular a positive and constructive answer to the second question would provide a much more experimentally-friendly method of producing high randomness rates as it would be able to tolerate natural imprecisions in the measurement device. 

\section{Application: How much cryptographically-secure randomness can you extract from $\rho_A$?}\label{sec:application}
Conditional entropies allow us to quantify the amount of information a potential eavesdropper may have about our measurement outcome. However, if we consider randomness generation as a cryptographic protocol then we look to satisfy some formal security definition (see~\cite{PR22} for details). If $\rho_{ZE}$ represents the final state of our randomness generation protocol, where $Z$ is our random string and $E$ is a potential adversary's system, then we say the protocol is $\epsilon$-secure (for some $\epsilon > 0$) if 
\begin{equation}\label{eq:sec}
	\frac12 \|\rho_{ZE} - \frac{\id_Z}{d_Z} \otimes \rho_E\|_1 \leq \epsilon
\end{equation}
where $d_Z$ is the dimension of the $Z$ register.\footnote{Here we are assuming the protocol never aborts.} I.e., the final state of the protocol $\rho_{ZE}$ should be basically indistinguishable from secret, uniform randomness.

For a general quantum state $\rho_A$, the $\rho_{XE}$ corresponding to the cq-state after measuring a POVM $\{M_x\}_x$ on $A$ will not satisfy the above definition. This is because the adversary will still, in general, have some partial knowledge about $X$. However, a procedure known as randomness extraction (or privacy amplification) allows us to extract the randomness from the large partially random $X$ system into a new $Z$ system that is $\epsilon$-close to uniformly random for any $\epsilon>0$ of our choosing (i.e. satisfies the security definition in~\eqref{eq:sec}). The amount of uniform randomness we can extract, i.e. the number of bits in the system $Z$, is quantified by certain conditional entropies. In particular we have the following result of~\cite{D21} which shows that the sandwiched Rényi entropies capture an achievable rate of randomness extraction. 
\begin{lemma}[Dupuis~\cite{D21}]\label{lem:dupuis}
	Let $\mathcal{F}$ be a family of $2$-universal hash functions from $X$ to $Z$. Then for any $\alpha \in (1,2]$ and cq-state $\rho_{XE} \in \D(XE)$
	\begin{equation}
		\frac12 \| \rho_{\cF(XS)E} - \frac{\id_Z}{d_Z} \otimes \rho_{SE}\|_1 \leq 2^{2/\alpha -2} 2^{\frac{\alpha-1}{\alpha} (\log(d_Z) - \swHu{\alpha}(X|E))}
	\end{equation}
	where $S$ is the seed used to sample the family. 
\end{lemma}
This lemma indicates that we can extract approximately $\swHu{\alpha}(X|E)$ uniformly random bits from $X$. In particular if we fix $\alpha$ and choose the output length of the extractor to be $\swHu{\alpha}(X|E) + \frac{\alpha}{\alpha-1} \log(\epsilon)$ then the output of the extractor is $\epsilon$-secure according to the above definition. This leads to the main result of this work.
\begin{theorem}[Main result]\label{thm:application}
	Let $\epsilon > 0$, $\rho_A \in \D(A)$ then it is possible to extract $l$ bits of $\epsilon$-secure randomness from $\rho_A$ for any $l\in \NN$ satisfying 
	\begin{equation}\label{eq:eps-rand}
		l < 2\log(d_A) - \min_{\alpha \in (1,2]} \left(H_{\frac{\alpha}{2\alpha-1}}(A) - \frac{\alpha}{\alpha-1} \log(\epsilon) \right)\,.
	\end{equation}
	Moreover, given $n \in \NN$ copies of $\rho_A$ it is possible to extract $l_n$ bits of $\epsilon$-secure randomness from $\rho_A^{\ox n}$ for any $l_n \in \NN$ satisfying
	\begin{equation}\label{eq:eps-rand-rates}
		\frac{l_n}{n} <  2\log(d_A) -  \min_{\alpha \in (1,2]} \left(H_{\frac{\alpha}{2\alpha-1}}(A) - \frac{\alpha}{n(\alpha-1)} \log(\epsilon) \right)\,.
	\end{equation}
\begin{proof}
	By Theorem~\ref{thm:int-rand-povm-sw} for any $\delta > 0$ and $\alpha \in (1,2]$ there exists a POVM $\{M_x\}_x$ such that the post-measurement state $\rho_{XE}$ satisfies 
	\begin{equation}\label{eq:sw-bound}
		\swHu{\alpha}(X|E) \geq 2 \log(d_A) - H_{\frac{\alpha}{2\alpha-1}}(A) - \delta\,.
	\end{equation}
	Thus, if we choose the output length of the extractor to be any $l \leq  2\log(d_A) - H_{\frac{\alpha}{2\alpha-1}}(A) - \delta + \frac{\alpha}{\alpha-1} \log(\epsilon)$, then by Lemma~\ref{lem:dupuis} we have
	\begin{equation}
\begin{aligned}
			\frac12 \| \rho_{\cF(XS)E} - \frac{\id_Z}{d_Z} \otimes \rho_{SE}\|_1 &\leq 2^{2/\alpha -2} 2^{\frac{\alpha-1}{\alpha} (\log(d_Z) - \swHu{\alpha}(X|E))} \\
			&\leq 2^{2/\alpha -2} 2^{\frac{\alpha-1}{\alpha} (2\log(d_A) - H_{\frac{\alpha}{2\alpha-1}}(A) -\delta + \frac{\alpha}{\alpha-1} \log(\epsilon) - \swHu{\alpha}(X|E))} \\
			&\leq 2^{2/\alpha -2} 2^{\frac{\alpha-1}{\alpha} ( \frac{\alpha}{\alpha-1} \log(\epsilon) )} \\
			&\leq \epsilon\,,
\end{aligned}
	\end{equation}
where on the second line we use the assumption $\log(d_Z) \leq 2\log(d_A) - H_{\frac{\alpha}{2\alpha-1}}(A)- \delta + \frac{\alpha}{\alpha-1} \log(\epsilon)$, on the third line we use~\eqref{eq:sw-bound} and on the final line we note that $2^{2/\alpha -2} \leq 1$ for all $\alpha \in (1,2]$. As this holds for all $\delta > 0$ the result follows. The second statement of the corollary follows from the same argument noting that $\maxintrinsicpovm{\HH}{\rho_A^{\otimes n}} = n \maxintrinsicpovm{\HH}{\rho_A} $.
\end{proof} 
\end{theorem}

\begin{remark}[Extraction with projective measurements]
	The above theorem could also be stated with the restriction of extracting randomness using projective measurements. In that case one just replaces $2 \log(d_A)$ with $\log(d_A)$. In particular we find that using projective measurements we can extract up to $l$-bits of $\epsilon$-secure randomness for any $l\in\NN$ satisfying  
	\begin{equation}
		l \leq \log(d_A) - \min_{\alpha \in (1,2]} \left(H_{\frac{\alpha}{2\alpha-1}}(A) - \frac{\alpha}{\alpha-1} \log(\epsilon) \right)\,.
	\end{equation}
\end{remark}

The statement in Theorem~\ref{thm:application} regarding extracting randomness from $n$ copies shows that we can approach the asymptotic rates of randomness extraction as $n \to \infty$. In particular, taking $\alpha = 1 + \frac{1}{\sqrt{n}}$ we see that for any $\epsilon > 0$ we have
\begin{equation}
	\frac{l_n}{n} < 2\log(d_A) - \swHu{\frac{1+\sqrt{n}}{2+\sqrt{n}}}(A) + \frac{1+\sqrt{n}}{n} \log(\epsilon) 
\end{equation}
and the right-hand-side tends to $\maxintrinsicpovm{H}{\rho_A} = 2\log(d_A) - H(A)$ as $n \to \infty$. Note this also gives an operational meaning to the bounds obtained in~\cite{MCSWFSA23}, they can be seen as quantifying the maximal amount of randomness that can be extracted from a quantum state using projective measurements in the asymptotic limit. 

\begin{remark}[Maximal randomness is not necessarily uniform]
	Suppose you want to extract as much $\epsilon$-secure randomness from $\rho_A$ as possible. Thus you find the $\alpha$ that maximizes~\eqref{eq:eps-rand-rates} and then you look for an extremal rank-one POVM $\{M_x\}_x$ that maximizes $\intrinsicpovm{\swHu{\alpha}}{\rho_A}{\{M_x\}_x}$. It is worth noting that the best choice of POVM may not be a one that produces a uniform distribution when measured on $\rho_A$ (or even close to uniform). Instead, as stated in Theorem~\ref{thm:int-rand-povm-sw}, the optimal measurement choice should produce a uniform distribution when measured on the state $\rho_A^{\alpha/(2\alpha-1)}/\tr{\rho_A^{\alpha/(2\alpha-1)}}$ and this does not imply that the measurement statistics you obtain for $\rho_A$ are uniform. Example~\ref{ex:opt-qubit} provides explicit instances of this, i.e., there exist extremal rank-one POVMs that produce uniform randomness on $\rho_A$ but which do not achieve $\maxintrinsicpovm{\HH}{\rho_A}$ whereas the POVM that achieves $\maxintrinsicpovm{\HH}{\rho_A}$ does not produce uniform randomness. 
\end{remark}

\begin{figure}
	\definecolor{mycolor2}{rgb}{0.00000,0.44700,0.74100}%
	\definecolor{mycolor4}{rgb}{0.85000,0.32500,0.09800}%
	\definecolor{mycolor3}{rgb}{0.92900,0.69400,0.12500}%
	\definecolor{mycolor1}{rgb}{0.49400,0.18400,0.55600}%
	\centering
	\begin{tikzpicture}[scale = 0.9]	
		\begin{axis}[%
			width=4in,
			height=2.8in,
			scale only axis,
			xmin=1e2,
			xmax=1e6,
			ymin=1.5,
			ymax=1.87,
			grid=major,
			xmode=log,
			ytick = {1.5, 1.55, 1.6, 1.65, 1.7, 1.75,1.8, 1.85},
			xtick = {1e2, 1e3, 1e4, 1e5, 1e6},
			yticklabels={1.5, 1.55, 1.6, 1.65, 1.7, 1.75,1.8, 1.85},
			axis background/.style={fill=white},
			xlabel={Number of samples ($n$)},
			ylabel={$\epsilon$-secure randomness rate (bits/$n$)},
			legend style={at={(0.44,0.99)},legend cell align=left, align=left, draw=white!15!black}
			]
			\addplot[color=black, line width=1.5pt, dashed] table[col sep=comma] {./data/ratesAsym.csv};
			\addlegendentry{Asymptotic rate};
			\addplot[color=mycolor1, line width=1.5pt] table[col sep=comma] {./data/ratesUopt4.csv};
			\addlegendentry{\eqref{eq:eps-rand-rates} with $\epsilon = 10^{-4}$};
			\addplot[color=mycolor1, line width=1.5pt, dotted, forget plot] table[col sep=comma] {./data/ratesDopt4.csv};
			\addplot[color=mycolor2, line width=1.5pt] table[col sep=comma] {./data/ratesUopt12.csv};
			\addlegendentry{\eqref{eq:eps-rand-rates} with $\epsilon = 10^{-12}$};
			\addplot[color=mycolor2, line width=1.5pt, dotted, forget plot] table[col sep=comma] {./data/ratesDopt12.csv};
			\addplot[color=mycolor4, line width=1.5pt] table[col sep=comma] {./data/ratesUopt20.csv};
			\addlegendentry{\eqref{eq:eps-rand-rates} with $\epsilon = 10^{-20}$};
			\addplot[color=mycolor4, line width=1.5pt, dotted, forget plot] table[col sep=comma] {./data/ratesDopt20.csv};
		\end{axis}
		\node[] at (7.5,1.75) {$\rho_A = \frac17 \begin{pmatrix}
				3 & \tfrac{1}{\sqrt{2}} & \tfrac{1}{\sqrt{2}} \\ \tfrac{1}{\sqrt{2}} & 2 & 1 \\ \tfrac{1}{\sqrt{2}} & 1 & 2
			\end{pmatrix}$};
	\end{tikzpicture}%
	\caption{A plot of the rate of $\epsilon$-secure randomness extractable from a memoryless source producing a state $\rho_A$ for various values of $\epsilon$ using~\eqref{eq:eps-rand-rates}. The dotted lines represent a lower bound on the rate of $\epsilon$-secure randomness extractable.}
	\label{fig:res}
\end{figure}

In Figure~\ref{fig:res} we exhibit the speed of convergence to the asymptotic rates of randomness extraction. We consider the state 
\begin{equation}
\rho_A = \frac17 \begin{pmatrix}
	3 & \tfrac{1}{\sqrt{2}} & \tfrac{1}{\sqrt{2}} \\ \tfrac{1}{\sqrt{2}} & 2 & 1 \\ \tfrac{1}{\sqrt{2}} & 1 & 2
\end{pmatrix}
\end{equation}
which has a spectrum $\{4/7,2/7,1/7\}$. We plot the rate of $\epsilon$-secure extractable randomness for $n$-uses of the source, for various values of $\epsilon$. This corresponds to computing~\eqref{eq:eps-rand-rates} by optimizing over $\alpha \in (1,2]$. For comparison, we also plot the following lower bound on the $\epsilon$-secure extractable randomness rates 
\begin{equation}\label{eq:eps-rand-rates-lb}
	2\log(d_A) -  \min_{\alpha \in (1,2]} \left(H_{\frac{1}{\alpha}}(A) - \frac{\alpha}{n(\alpha-1)} \log(\epsilon) \right)
\end{equation}
which corresponds to applying~\cite{D21} with $\swHd{\alpha}(X|E)$. This captures the loss in rates that occurs when we use $\swHd{\alpha}$ instead of $\swHu{\alpha}$, which in practical security proofs is a common technique due to the difficulty 
of computing~$\swHu{\alpha}$~\cite{HB24}.

In the plot, for each value of $\epsilon$ considered, we observe the convergence to the asymptotic rate as $n \to \infty$ (including for the lower bounds). In fact for $n>10^4$ we see that there is little difference between~\eqref{eq:eps-rand-rates} and its lower bound~\eqref{eq:eps-rand-rates-lb}, as for this value of $n$ we have $\alpha \approx 1$ which implies that $\swHu{\alpha} \approx \swHd{\alpha}$. However, for smaller $n$, which is also of practical relevance, we see a significant difference between using $\swHu{\alpha}$ or $\swHd{\alpha}$. This highlights the importance of developing general methods to compute and optimize $\swHu{\alpha}(X|E)$ for the purpose of practical and tight finite-size security proofs. We also note that for small $n$, at some point the optimal $\alpha$ hits the barrier of $\alpha = 2$ in the minimization, in the plot this is evident with the kinks in the various curves. This hints at the potential for a possible improvement in the result of~\cite{D21} which would allow for the optimization to be over all $\alpha > 1$. 

\section{Conclusion}

In this work we derived simple to compute expressions for the quantity of $\epsilon$-secure randomness that is extractable from a given quantum state using both projective measurements and POVMs. Our adversarial model was general allowing for the adversary to have quantum side-information on both the source and the measurement device. Overall this provides a simple metric for the usefulness of a given source for QRNG protocols and acts as an easy to compute benchmarking tool for their performance and security proofs.

To derive these results we defined the maximal intrinsic $\HH$-randomness of a quantum state, generalizing the intrinsic randomness problem found in~\cite{MCSWFSA23}. We were able to derive closed form solutions to this problem for all quantum states and the entire families of operationally relevant sandwiched Rényi entropies. However our proofs of these results raise some interesting questions. In particular, our proof of the maximal randomness extractable using POVMs relies on a limiting argument. That is, we prove the existence of a sequence of POVMs that approach the maximal intrinsic $\HH$-randomness of the given state. However, the limiting POVM in our construction does not achieve the maximal intrinsic $\HH$-randomness rate, demonstrating a discontinuity in the intrinsic randomness measure. It would be interesting to understand if the maximal intrinsic randomness is always achievable with some POVM. In particular, we showed that they are achievable for all qubit states, however for states of higher dimensions we do not have explicit constructions.

Additionally, the discontinuity encountered in the limiting argument prompts a pertinent practical question. Do there exist POVMs that can produce high rates of extractable randomness in a robust manner? In particular, physical implementations will naturally implement a noisy version of the chosen POVM. Due to the existence of discontinuities, it is unclear a priori if a noisy version of a given POVM will still have a high rate of extractable randomness. Searching for robust POVM constructions that enable the extraction of large amounts of randomness would be of significant interest for optimizing the design of practical QRNGs.

Finally, we note that our scheme to extract the secure randomness from a quantum state involves Alice first measuring her quantum system and then applying a randomness extraction protocol. We note that a more direct method has been considered in~\cite{BFW13} wherein so-called quantum-to-classical randomness extractors are introduced. This can be seen as merging of the two steps in our scheme into a single step. In particular, in the scheme of~\cite{BFW13} a family of quantum to classical channels $\cE_i : L(A) \to L(X)$ is defined such that $\cE_i$ is selected randomly from the family and applied to the source $A$. This can be physically more difficult to implement and hence it is not the approach taken by most QRNGs. However, it is possible that higher rates of extraction could be achieved under this more general scheme. We leave the investigation of quantum to classical randomness extractors to future work.  

\begin{acknowledgements}
	The authors thank Fionnuala Curran for pointing out an error in a previous version of the manuscript. The authors acknowledge funding from FranceQCI, funded by the Digital Europe program,  DIGITAL-2021-QCI-01, project no. 101091675. P.B. acknowledges funding from the European Union’s Horizon Europe research and innovation programme under the project “Quantum Secure Networks Partnership” (QSNP, grant agreement No. 101114043).
\end{acknowledgements}

\bibliographystyle{quantum}
\bibliography{tradeoff.bib}

\appendix

\section{Proof of Theorem~\ref{thm:int-rand-proj-sw}}\label{app:sw}
In this section we prove the closed form expressions for $\maxintrinsicpvm{\swHu{\alpha}}{\rho_A}$ and $\maxintrinsicpvm{\swHd{\alpha}}{\rho_A}$. We begin with the following lemma. 

\begin{lemma}\label{lem:div-expr}
    Let $\alpha \in [1/2,1) \cup (1, \infty)$ and let $\beta > 0$. Let $\rho_{AE}$ be a pure state, $\{M_x\}_x$ be a rank-one measurement and define the cq-state $\rho_{XE} = \sum_x \outer{x} \otimes \rho_{E,x}$ where $\rho_{E,x} = \ptr{A}{(M_x \otimes \id_E) \rho_{AE}}$. Then
    \begin{equation}
    	\swD{\alpha}\left(\rho_{XE}\middle\|\id_A\otimes\frac{\rho_E^\beta}{\tr{\rho_E^\beta}}\right) =
    	\frac{1}{\alpha-1} \log \left( \frac{\sum_x  \tr{\rho_A^{\frac{\beta(1-\alpha)}{\alpha} + 1} M_x }^\alpha}{\tr{\rho_A^\beta}^{1-\alpha}} \right)\,.
    \end{equation}
\end{lemma}
\begin{proof}
	As $\rho_{XE}$ is a classical-quantum state, for any density matrix $\sigma_E$ satisfying $\supp(\rho_{E,x}) \subseteq \supp(\sigma_E)$ for all $x$ we have 
	\begin{equation}\label{eq:Da-full}
		\swD{\alpha}(\rho_{XE} \| \id_X \otimes \sigma_E) = \frac{1}{\alpha-1} \log \left( \sum_x \tr{ \left(\sigma_E^{\frac{1-\alpha}{2\alpha}} \rho_{E,x} \sigma_E^{\frac{1-\alpha}{2\alpha}} \right)^\alpha} \right) \,.
	\end{equation}
	Now by the Schmidt decomposition we can write $\rho_{AE} = \outer{\psi}$ with $\ket{\psi} = \sum_{i} \sqrt{\lambda_i} \ket{v_i}_A \otimes \ket{w_i}_E$. As $M_x = \outer{m_x}$ is a rank-one measurement operator we see that 
	\begin{equation}
            \label{eq:cq-state-pvm}
		\begin{aligned}
			\rho_{E,x} &= \ptr{A}{(M_x \otimes \id) \rho_{AE}} \\
			&= \sum_{ij} \sqrt{\lambda_i \lambda_j} \inner{m_x}{v_i}\inner{v_j}{m_x} \outer{w_i}{w_j} \\
			&= \outer{\psi_x}_E
		\end{aligned}
	\end{equation}
	where we have defined the (not necessarily normalized) vector $\ket{\psi_x}_E = \sum_i \sqrt{\lambda_i} \inner{m_x}{v_i} \ket{w_i}$. Thus, as $\rho_{E,x}$ is always rank-one, \eqref{eq:Da-full} simplifies to 
	\begin{equation}
	\swD{\alpha}(\rho_{XE} \| \id_X \otimes \sigma_E) = \frac{1}{\alpha-1} \log \left( \sum_x \left(\bra{\psi_x} \sigma_E^{\frac{1-\alpha}{\alpha}} \ket{\psi_x}\right)^\alpha \right) \,.
	\end{equation}
	Now substituting in $\sigma_E = \frac{\rho_E^\beta}{\tr{\rho_E^\beta}}$ the right-hand-side becomes
	\begin{equation}
		\begin{aligned}
			\frac{1}{\alpha-1} \log \left( \frac{\sum_x \left(\bra{\psi_x} \rho_E^{\frac{\beta(1-\alpha)}{\alpha}} \ket{\psi_x}\right)^\alpha}{\tr{\rho_E^\beta}^{1-\alpha}} \right) &= \frac{1}{\alpha-1} \log \left( \frac{\sum_x \left(\bra{\psi_x} \left(\sum_i \lambda_i^{\frac{\beta(1-\alpha)}{\alpha}} \outer{w_i} \right) \ket{\psi_x}\right)^\alpha}{\tr{\rho_E^\beta}^{1-\alpha}} \right) \\
			&= \frac{1}{\alpha-1} \log \left( \frac{\sum_x  \left(\sum_i \lambda_i^{\frac{\beta(1-\alpha)}{\alpha} + 1} |\inner{v_i}{m_x}|^2 \right)^\alpha}{\tr{\rho_E^\beta}^{1-\alpha}} \right) \\
			&= \frac{1}{\alpha-1} \log \left( \frac{\sum_x  \tr{\rho_A^{\frac{\beta(1-\alpha)}{\alpha} + 1} M_x }^\alpha}{\tr{\rho_A^\beta}^{1-\alpha}} \right) \\
		\end{aligned}
	\end{equation}
	where on the final line we have noted that $\tr{\rho_E^\beta} = \tr{\rho_A^\beta}$ as $\rho_{AE}$ is pure. 
\end{proof}

The above lemma allows us to simplify the expressions of the sandwiched-Rényi entropies. 

\begin{lemma}\label{lem:sw-ent-form}
	Let $\rho_{AE}$ be a pure state, $\{M_x\}_x$ be a rank-one measurement and define the cq-state $\rho_{XE} = \sum_x \outer{x} \otimes \rho_{E,x}$ where $\rho_{E,x} = \ptr{A}{(M_x \otimes \id) \rho_{AE}}$. Then for $\alpha \in [1/2,1) \cup (1, \infty)$ we have 
	\begin{equation}
		\swHd{\alpha}(X|E) = H_\alpha(X') - H_{1/\alpha}(A)
	\end{equation} 
	where $X'$ is a random variable with a distribution $p(x) = \tr{\rho_A^{1/\alpha} M_x}/\tr{\rho_A^{1/\alpha}}$. Additionally, for $\alpha > 1$ we have
	\begin{equation}
		\swHu{\alpha}(X|E) = H_\alpha(X'') - H_{\frac{\alpha}{2\alpha-1}}(A)
	\end{equation}
	where $X''$ is a random variable with a distribution $p(x) = \tr{\rho_A^{\tfrac{\alpha}{2\alpha-1}} M_x}/ \tr{\rho_A^{\tfrac{\alpha}{2\alpha-1}}}$.
	\begin{proof}
		Let us begin with $\swHd{\alpha}(X|E)$. Taking $\beta=1$ in~Lemma~\ref{lem:div-expr}, we have
		\begin{equation}
			\begin{aligned}
				\swHd{\alpha}(X|E)_{\rho_{XE}} &= \frac{1}{1-\alpha}\log\left(\sum_x\left(\tr{\rho_A^{\frac1\alpha}M_x}\right)^{\alpha}\right)\\
				&= \frac{1}{1-\alpha}\log\left(\tr{\rho_A^{\frac1\alpha}}^{\alpha}\sum_x\left(\frac{\tr{\rho_A^{\frac1\alpha}M_x}}{\tr{\rho_A^{\frac1\alpha}}}\right)^{\alpha}\right)\\
				&= \frac{\alpha}{1-\alpha}\log\left(\tr{\rho_A^{\frac1\alpha}}\right)+\frac{1}{1-\alpha}\log\left(\sum_x\left(\frac{\tr{\rho_A^{\frac1\alpha}M_x}}{\tr{\rho_A^{\frac1\alpha}}}\right)^{\alpha}\right) \\
				&= - H_{1/\alpha}(A)+ H_{\alpha}(X') 
			\end{aligned}
		\end{equation}
		where on the fourth line we introduced the classical random variable $X'$ with distribution $p(x) = \tr{\rho_A^{1/\alpha}M_x}/\tr{\rho_A^{1/\alpha}}$ and noted that the term was a classical Rényi entropy.
		
		Now we prove the case of $\swHu{\alpha}(X|E)$.  Using \cite[Corollary~4]{TBH14} and the above result for $\swHd{\alpha}(X|E)$, we have for any $\alpha >1$
		\begin{equation}
			\begin{aligned}
				\swHu{\alpha}(X|E)_{\rho_{XE}} &\leq \swHd{2-\frac1\alpha}(X|E)_{\rho_{XE}}\\
				&= H_\alpha(X'')- H_{\frac{\alpha}{2\alpha-1}}(A)\,,
			\end{aligned}
		\end{equation}
		where $X''$ is a classical random variable with a distribution
		\begin{equation}
			p(x) = \tr{\rho_A^{\tfrac{\alpha}{2\alpha-1}} M_x}/ \tr{\rho_A^{\tfrac{\alpha}{2\alpha-1}}}\,.
		\end{equation}
		However, if we apply Lemma~\ref{lem:div-expr} with $\beta = \alpha/(2\alpha-1)$ then we find also the lower bound 
		\begin{equation}
			\begin{aligned}
				\swHu{\alpha}(X|E) &\geq \frac{1}{1-\alpha}
				\log\left(\sum_x\frac{\tr{\rho_A^{\frac{\alpha}{2\alpha-1}}M_x}^{\alpha}}{\tr{\rho_A^{\frac{\alpha}{2\alpha-1}}}^{1-\alpha}}\right) \\
				&= \frac{1}{1-\alpha}
				\log\left(\sum_x\frac{\tr{\rho_A^{\frac{\alpha}{2\alpha-1}}M_x}^{\alpha}}{\tr{\rho_A^{\frac{\alpha}{2\alpha-1}}}^{\alpha}}\right) + \frac{2\alpha -1}{1-\alpha} \log\left( \tr{\rho_A^{\frac{\alpha}{2\alpha-1}}}\right) \\
				&= H_\alpha(X'') - H_{\frac{\alpha}{2\alpha-1}}(A)\,,
			\end{aligned}
		\end{equation}
		which completes the proof.
	\end{proof}
\end{lemma}

Using the above lemma we can now find our closed form expressions.
\begin{lemma}\label{lem:swHd}
	Let $\alpha \in [1/2,1) \cup (1, \infty)$ and let $\rho_A \in \D(A)$. Then
   	\begin{equation}
   		\maxintrinsicpvm{\swHd{\alpha}}{\rho_A}=\log(d_A)-H_{1/\alpha}(A)
   	\end{equation}
   	where $d_A$ is the dimension of system $A$. Furthermore, this quantity is achieved if and only if the rank-one PVM \(\left\{M_x\right\}_x\) satisfies
    \begin{equation}
        \tr{\rho_A^{\frac1\alpha}M_x} = \frac{\tr{\rho_A^{\frac1\alpha}}}{d_A}
    \end{equation}
    for all \(x\in\left[d_A\right]\).
\end{lemma}
\begin{proof}
	By Lemma~\ref{lem:sw-ent-form} and Lemma~\ref{lem:rank1} we know that 
	\begin{equation}
		\maxintrinsicpvm{\swHd{\alpha}}{\rho_A} = \sup_{\mathrm{rank-one~PVMs}} \swHd{\alpha}(X|E) = \sup_{\mathrm{rank-one~PVMs}} H_\alpha(X') - H_{1/\alpha}(A)
	\end{equation}
	 where $X'$ is a random variable with distribution $p(x) = \tr{\rho_A^{1/\alpha} M_x}/\tr{\rho_A^{1/\alpha}}$. Since the only $H_\alpha(X')$ depends on the PVM $\{M_x\}_x$, we look to maximize this term. We know that $H_{\alpha}(X') = \log(d_{A})$ is maximal iff there exists a rank-one projective measurement \(\left\{M_x\right\}_x\) such that
    \begin{equation}
        \frac{\tr{\rho_A^{\frac1\alpha}M_x}}{\tr{\rho_A^{\frac1\alpha}}} = \frac{1}{d_A}
    \end{equation}
    for all \(x\in\left[d_A\right]\). To complete the proof, we note that this condition is achieved by taking $\{M_x\}_x$ to be the projectors onto a basis which is mutually unbiased with the eigenbasis of $\rho_A$.
\end{proof}

We can now use the result for $\maxintrinsicpvm{\swHd{\alpha}}{\rho_A}$ to derive a result for $\maxintrinsicpvm{\swHu{\alpha}}{\rho_A}$.

\begin{lemma}\label{lem:swHu}
	Let $\alpha > 1$ and let $\rho_A \in D(A)$. Then
    \begin{equation}
    	\maxintrinsicpvm{\swHu{\alpha}}{\rho_A}=\log(d_A)-H_{\frac{\alpha}{2\alpha-1}}\left(A\right)
    \end{equation}
    where $d_A$ is the dimension of system $A$. Furthermore, this quantity is achieved if and only if the rank-one measurement \(\left\{M_x\right\}_x\) satisfies
    \begin{equation}
        \tr{\rho_A^{\frac{\alpha}{2\alpha-1}}M_x} = \frac{\tr{\rho_A^{\frac{\alpha}{2\alpha-1}}}}{d_A}
    \end{equation}
    for all \(x\in\left[d_A\right]\).
\end{lemma}
\begin{proof}
    	By the same logic as the previous lemma, we can use Lemma~\ref{lem:sw-ent-form} and Lemma~\ref{lem:rank1} to show that 
    \begin{equation}
    	\maxintrinsicpvm{\swHu{\alpha}}{\rho_A} = \sup_{\mathrm{rank-one~PVMs}} \swHu{\alpha}(X|E) = \sup_{\mathrm{rank-one~PVMs}} H_\alpha(X') - H_{\frac{\alpha}{2\alpha-1}}(A)
    \end{equation}
    where $X'$ is a random variable with distribution $p(x) = \tr{\rho_A^{\frac{\alpha}{2\alpha-1}} M_x}/\tr{\rho_A^{\frac{\alpha}{2\alpha-1}}}$. Since the only $H_\alpha(X')$ depends on the PVM $\{M_x\}_x$, we look to maximize this term. We know that $H_{\alpha}(X') = \log(d_{A})$ is maximal iff there exists a rank-one projective measurement \(\left\{M_x\right\}_x\) such that
    \begin{equation}
    	\frac{\tr{\rho_A^{\frac{\alpha}{2\alpha-1}}M_x}}{\tr{\rho_A^{\frac{\alpha}{2\alpha-1}}}} = \frac{1}{d_A}
    \end{equation}
    for all \(x\in\left[d_A\right]\). To complete the proof, we note that this condition is achieved by taking $\{M_x\}_x$ to be the projectors onto a basis which is mutually unbiased with the eigenbasis of $\rho_A$.
\end{proof}

\begin{remark}[Intrinsic randomness for $H$ and $H_{\min}$]
	The proof of Lemma~\ref{lem:swHu} also extends to the limit $\alpha \to 1$, recovering
	\begin{equation}
		\maxintrinsicpvm{H}{\rho_A} = \log(d_A) - H(A)\,,
	\end{equation}
	and to the limit $\alpha \to \infty$, recovering
	\begin{equation}
		\maxintrinsicpvm{H_{\min}}{\rho_A} = \log(d_A) - H_{1/2}(A)\,,
	\end{equation}
	which were proven in~\cite{MCSWFSA23} using different techniques.
\end{remark}

\section{Proof of Theorem~\ref{thm:int-rand-povm-sw}}\label{app:max-sw-povm}

In this section we will develop all of the necessary results to prove Theorem~\ref{thm:int-rand-povm-sw}. We begin with some preliminary results.

\begin{lemma}[Convexity of exponential intrinsic randomness]\label{lem:convexity}
	Let $\alpha \in [\tfrac12,1) \cup (1,\infty)$, let \(\HH\) be either \(\swHd{\alpha}\) or \(\swHu{\alpha}\) and let $\rho_A \in \D(A)$. Then for any two POVMs $\{M_x\}_x$, $\{N_x\}_x$ and $\lambda \in [0,1]$ we have
	\begin{equation}
		2^{\intrinsicpovm{\HH}{\rho_A}{\{\lambda\,M_x + (1-\lambda)\, N_x\}_x}} \leq \lambda \,2^{\intrinsicpovm{\HH}{\rho_A}{\{M_x\}_x}} + (1-\lambda)\,2^{\intrinsicpovm{\HH}{\rho_A}{\{N_x\}_x}}\,.
	\end{equation}
	\begin{proof}
		Let \(\varepsilon>0\). Let $(\{P_x\}_x, \sigma_B)$ be a dilation for $\intrinsicpovm{\HH}{\rho_A}{\{M_x\}_x}$ and let $(\{Q_x\}_x, \tau_{B'})$ be a dilation for $\intrinsicpovm{\HH}{\rho_A}{\{N_x\}_x}$ satisfying
		\begin{equation}
			\HH(X|E)_{\sigma_{XE}}\leq\intrinsicpovm{\HH}{\rho_A}{\{M_x\}_x}+\varepsilon
		\end{equation}
		and 
		\begin{equation}
			\HH(X|E)_{\tau_{XE}}\leq\intrinsicpovm{\HH}{\rho_A}{\{N_x\}_x}+\varepsilon.
		\end{equation}        
		Note that we may assume \(\dim(B)=\dim\left(B'\right)\) without loss of generality. To see this note that if \(\dim(B)<\dim\left(B'\right)\) then we can embed $B$ into $B'$ using an isometry $V: B \to B'$ which transforms the dilation to $(\{V P_x V^{\dagger}\}_x, V \sigma V^{\dagger})$ where to complete the projective measurement we add a new outcome $P_{\perp} = \id_{AB'} - \id_A \otimes VV^{\dagger}$. However by construction this new measurement operator satisfies $\tr{P_{\perp}(\id_A \otimes V \sigma V^{\dagger})} = 0$. This construction yields a new dilation for $\{M_x\}_x$ but it does not change $\HH(X|E)$ as the post-measurement state remains unchanged. Hence we may assume the dilating systems are the same for both POVMs. 
		
		Hence we have two states $\sigma_{AB}$ and $\tau_{AB}$ such that $\ptr{B}{\sigma_{AB}} = \rho_A$ and $\ptr{B}{\tau_{AB}} = \rho_A$ and  $(\{P_x\}_x, \sigma_B)$ is a dilation of $\{M_x\}_x$ and $(\{Q_x\}_x, \tau_B)$ is a dilation of $\{N_x\}_x$. Now for any $\lambda \in[0,1]$, consider a new state 
		\begin{equation}
			\omega_{ABF} = \lambda\, \sigma_{AB} \otimes \outer{0}_F + (1-\lambda)\,\tau_{AB} \outer{1}_F
		\end{equation}
		and projective measurement on $ABF$ defined as 
		\begin{equation}
			R_x =  P_x \otimes \outer{0}_F + Q_x \otimes \outer{1}_F\,.
		\end{equation}
		Then a direct calculation shows that this pair forms a dilation of the mixed measurement $\{\lambda M_x + (1-\lambda) N_x\}_x$. So consider the state
		\begin{equation}
			\rho_{ABEFF'} = \lambda \sigma_{ABE} \otimes \outer{00}_{FF'} + (1-\lambda) \tau_{ABE} \otimes \outer{11}_{FF'}
		\end{equation}
		where $\sigma_{ABE}$ and $\tau_{ABE}$ are pure and let $\rho_{ABEFF'G}$ be a purification of $\rho_{ABEFF'}$ we assume that $EF'G$ are held by Eve and $BF$ are held by the measurement device. 
		
		Now by definition and strong subadditivity we have
		\begin{equation}
			2^{\intrinsicpovm{\HH}{\rho_A}{\{\lambda\, M_{x} + (1-\lambda)\, N_x \}_x}} \leq 2^{\HH(X|EF'G)_{\rho_{XEF'G}}} \leq 2^{\HH(X|EF')_{\rho_{XEF'}}}.
		\end{equation}
		Thus, as $F'$ is a classical flag system we have from Lemma~\ref{lem:convexitypowerof2}
		\begin{equation}
			\begin{aligned}
				2^{\intrinsicpovm{\HH}{\rho_A}{\{\lambda\, M_{x} + (1-\lambda)\, N_x \}_x}}&\leq \lambda\, 2^{\HH(X|E)_{\sigma_{XE}}} + (1-\lambda)\, 2^{\HH(X|E)_{\tau_{XE}}} \\
				&\leq \lambda\,2^{\intrinsicpovm{\HH}{\rho_A}{\{M_x\}_x} + \varepsilon} + (1-\lambda)\,2^{\intrinsicpovm{\HH}{\rho_A}{\{N_x\}_x}+ \varepsilon}.
			\end{aligned}
		\end{equation}
		The result then follows from the fact that this holds for all $\epsilon > 0$.
	\end{proof}
\end{lemma} 

We now show that if the probe state of the dilation $\rho_B$ is not rank-one then we can always find a rank-one probe state that implements the same extremal POVM.  

\begin{lemma}[Adapted from {\cite[Proposition~1]{DCZM23}}]
	\label{lem:mixture-of-probe-strategy}
	Let \(\left\{M_x\right\}_{x}\) be an extremal POVM on \(A\) and let \(\left(\left\{P_x\right\}_x,\sigma_B\right)\) be a corresponding dilation.  Let $\sum_yp_y\,\outer{v_y}_B$ be the spectral decomposition of $\sigma_B$ with $\{\ket{v_y}\}_y$ forming an orthonormal basis. Then for all \(y\) such that $p_y > 0$ we have \(\left(\left\{P_x\right\}_x,\outer{v_y}_B\right)\) is also a dilation of $\{M_x\}_x$.
\end{lemma}
\begin{proof}
	From the consistency condition for dilations, we have that for all \(x\)
	\begin{equation}
		\begin{aligned}
			M_x&=\ptr{B}{P_x\left(\id_A\otimes\sigma_B\right)}\\
			&=\sum_yp_y\ptr{B}{P_x\left(\id_A\otimes\outer{v_y}_B\right)}.
		\end{aligned}
	\end{equation}
	Let us denote \(N_x^y=\ptr{B}{P_x\left(\id_A\otimes\outer{v_y}_B\right)}\). Note that we have for all \(y\)
	\begin{equation}
		\begin{aligned}
			\sum_{x}N_x^y &= \sum_{x}\ptr{B}{P_x\left(\id_A\otimes\outer{v_y}_B\right)}\\
			&= \ptr{B}{\left(\sum_{x}P_x\right)\left(\id_A\otimes\outer{v_y}_B\right)}\\
			&= \id_A.
		\end{aligned}
	\end{equation}
	Furthermore, by positivity of the partial trace we also have that $N_{x}^y$ is positive semidefinite and therefore it is a POVM on \(A\). In particular, since \(\left\{M_x\right\}_{x}\) is extremal and since we have
	\begin{equation}
		M_x=\sum_yp_yN^y_x
	\end{equation}
	it means that \(N^y_x=M_x\) for all \(y\) such that $p_y > 0$. Thus, for all \(y\), the consistency condition is satisfied for \(\left(\left\{P_x\right\}_x,\outer{v_y}_B\right)\), which means that it is a valid dilation of $\{M_x\}_x$.
\end{proof}

Using this result we can prove Lemma~\ref{lem:separable-probe-state} from the main text which we repeat here for clarity. Loosely, it shows that for extremal rank-one POVMs, the entanglement with the measurement device carries no useful information for Eve.

\begin{lemma}
	\label{lem:separable-probe-state-app}
	Let $\left\{M_x\right\}_x$ be an extremal rank-one POVM on a system $A$, let $\left(\left\{P_x\right\}_x, \rho_B\right)$ be any dilation of $\{M_x\}_x$ and let $\rho_{ABE} \in \D(ABE)$ be any pure state with marginals $\rho_A$ and $\rho_B$. Then the post-measurement state
	\begin{equation}
		\rho_{XE} = \sum_x \outer{x} \otimes \ptr{AB}{(P_x \otimes \id_E) \rho_{ABE}}
	\end{equation}
	can be written as 
	\begin{equation}
		\rho_{XE} = \sum_x \outer{x} \otimes \ptr{A}{(M_x \otimes \id_E)\rho_{AE}}\,.
	\end{equation}
	\begin{proof}
		Consider the spectral decomposition $\rho_B = \sum_k \lambda_k \outer{\tau_k}$. As $\{\ket{\tau_k}\}_k$ is an orthonormal basis, we can write
		\begin{equation}
			P_x = \sum_{ijkl} p_{ijkl} \outer{i}{j}_A \otimes \outer{\tau_k}{\tau_l}_B\,.
		\end{equation}
		for some coefficients $p_{ijkl} \in \CC$. By Lemma~\ref{lem:mixture-of-probe-strategy} we know that $\left(\left\{P_x\right\}_x,\outer{\tau_k}\right)$ is also a dilation for $\{M_x\}_x$ whenever $\lambda_k >0$ and so
		\begin{equation}
			M_x = (\id_A \otimes \bra{\tau_k}) P_x (\id_A \otimes \ket{\tau_k}) = \sum_{ij} p_{ijkk} \outer{i}{j},
		\end{equation}
		for any $k$ such that $\lambda_k$ > 0. However, $M_x = m_x \outer{v}$ is also rank-one, thus we can extend the vector $\ket{v}$ to an orthonormal basis $\{\ket{v_i}\}_i$ for $A$ (where we define $\ket{v_0} = \ket{v}$). This means that we can write
		\begin{equation}
			P_x = m_x \outer{v_0} \otimes \sum_{k:\lambda_k > 0} \outer{\tau_k} + \sum_{ijkl: k \neq l, \lambda_k, \lambda_l > 0} q_{ijkl} \outer{v_i}{v_j} \otimes \outer{\tau_k}{\tau_l} + R
		\end{equation}
		for some coefficients $q_{ijkl} \in \CC$ and $R$ is a matrix such that $(\id \otimes \bra{\tau_k})R (\id \otimes \ket{\tau_l}) = 0$ whenever $\lambda_k >0$ and $\lambda_l > 0$. We will show that if $i \neq 0$ and $j\neq 0$ then $q_{ijkl} = 0$. Suppose not and without loss of generality take $i=0$ and $j \neq 0$ (a similar argument holds for other cases) then consider the vector $\ket{w} = \alpha \ket{v_0}\otimes \ket{\tau_k} + \beta \ket{v_j} \otimes \ket{\tau_l}$ where $\lambda_k,\lambda_l > 0$, then as $P_x \succeq 0$ we should have
		\begin{equation}
			0 \leq \bra{w} P_x \ket{w} = |\alpha|^2 m_x + 2 \Re[\overline{\alpha} \beta q_{0jkl}] 
		\end{equation}
		it is clear that there is a choice of $\beta$ for which the RHS is negative if $q_{0jkl} \neq 0$. Thus we must have $q_{0jkl}=0$. Overall this implies that 
		\begin{equation}
			P_x = M_x \otimes \sum_{k,l, \lambda_k,\lambda_l >0} c_{klx} \outer{\tau_k}{\tau_l} + R
		\end{equation}
		for some coefficients $c_{klx} \in \CC$ which may a priori depend on $x$. We can now go one step further and show that $c_{klx} = 0$ when $k \neq l$. Note that $\sum_x P_x = \id$, thus if $k \neq l$ we have 
		\begin{equation}
			\begin{aligned}
				0 &= \sum_x (\id_A \otimes \bra{\tau_k}) P_x (\id_A \otimes \ket{\tau_l}) \\
				&= \sum_{x} c_{klx} M_x
			\end{aligned}
		\end{equation}
		but as $\{M_x\}_x$ is a rank-one extremal POVM, its elements are all linearly independent and hence we must have $c_{klx} = 0$ for all $k \neq l$. This further simplifies $P_x$ to
		\begin{equation}
			P_x = M_x \otimes \Pi_{\sigma} + R
		\end{equation}
		where $\Pi_\sigma = \sum_{k: \lambda_k > 0} \outer{\tau_k}$ is the projector onto the support of $\sigma$. 
		
		Finally, consider 
		\begin{equation}
			\begin{aligned}
				\ptr{AB}{(P_x \otimes \id_E) \rho_{ABE} } &= \ptr{AB}{((M_x \otimes\Pi_{\sigma} + R)\otimes \id_E) \rho_{ABE} } \\
				&= \ptr{AB}{(M_x \otimes\Pi_{\sigma}\otimes \id_E) \rho_{ABE} } \\
				&= \ptr{AB}{(M_x \otimes\id_B\otimes \id_E) \rho_{ABE} } \\
				&= \ptr{A}{M_x \rho_{AE}}\,.
			\end{aligned}
		\end{equation}
	\end{proof}
\end{lemma}

By Lemma~\ref{lem:extremal_povm_is_optimal} we know that 
\begin{equation}
	\maxintrinsicpovm{\HH}{\rho_A} = \sup_{\substack{\mathrm{Extremal~rank-one}\\ \{M_x\}_x}} \intrinsicpovm{\HH}{\rho_A}{\{M_x\}_x}\,.
\end{equation}
By the previous lemma we can further simplify this to
\begin{equation}\label{eq:maxintpovm-reduced}
	\begin{aligned}
		\maxintrinsicpovm{\HH}{\rho_A} = \sup& \quad \HH(X|E) \\
		\mathrm{s.t.}& \quad \rho_{XE} = \sum_x \outer{x} \otimes \ptr{A}{(M_x \otimes \id) \outer{\rho}_{AE}} \\
		& \quad \{M_x\}_x \text{ is an extremal rank-one POVM}
	\end{aligned}
\end{equation}
The following theorem establishes a closed form expression for~\eqref{eq:maxintpovm-reduced}.

\begin{theorem}\label{thm:maxint-app}
		For $\alpha \in [\tfrac12,1) \cup (1,\infty)$ we have
		\begin{equation}
			\maxintrinsicpovm{\swHd{\alpha}}{\rho_A} = 2 \log(d_A) - H_{1/\alpha}(A)\,.
		\end{equation}
		Furthermore, this quantity is achieved by an extremal rank-one POVM $\{M_x\}_x$ if and only if
		\begin{equation}
			\tr{\rho_A^{\frac1\alpha}M_x}=\frac{\tr{\rho_A^{\frac1\alpha}}}{d_A^2}\,
		\end{equation}
		for all $x \in [d_A^2]$. Similarly, for $\alpha > 1$ we have
		\begin{equation}
			\maxintrinsicpovm{\swHu{\alpha}}{\rho_A} = 2 \log(d_A) - H_{\frac{\alpha}{2\alpha-1}}(A)
		\end{equation}
		with this quantity being achieved by an extremal rank-one POVM $\{M_x\}_x$ if and only if
		\begin{equation}
			\tr{\rho_A^{\frac{\alpha}{2\alpha-1}}M_x}=\frac{\tr{\rho_A^{\frac{\alpha}{2\alpha-1}}}}{d_A^2}\,
		\end{equation}
		for all $x \in [d_A^2]$.
	\begin{proof}
		Let \(\HH\) be either \(\swHd{\alpha}\) or \(\swHu{\alpha}\) for \(\alpha\) in the respective ranges as defined in the statement of the theorem. By the previous arguments we have that 
		\begin{equation}
		\maxintrinsicpovm{\HH}{\rho_A} = \sup_{\substack{\mathrm{Extremal~rank-one}\\ \{M_x\}_x}} \HH(X|E)\,.
		\end{equation}
		As the cq-state $\rho_{XE}$ is a generated by a rank-one measurement $\{M_x\}_x$ on a pure state, we can apply Lemma~\ref{lem:sw-ent-form} to get, for $\alpha \in (1/2,1)\cup(1,\infty)$,
		\begin{equation}
			\swHd{\alpha}(X|E) = H(X') - H_{1/\alpha}(A) \qquad \text{with} \qquad p(x') = \frac{\tr{\rho_A^{1/\alpha} M_x}}{\tr{\rho_A^{1/\alpha}}} 
		\end{equation} 
		and for $\alpha > 1$
		\begin{equation}
			\swHu{\alpha}(X|E) = H(X'') - H_{\frac{\alpha}{2\alpha-1}}(A) \qquad \text{with} \qquad p(x'') = \tr{\rho_A^{\tfrac{\alpha}{2\alpha-1}} M_x}/ \tr{\rho_A^{\tfrac{\alpha}{2\alpha-1}}}\,.
		\end{equation}
		Thus the remainder of the proof entails maximizing $H(X')$ and $H(X'')$ over the set of extremal rank-one POVMs. In particular we look for an extremal rank-one  POVM which generates a uniform distribution for $X'$ and $X''$ which would establish that $H(X')= H(X'') = 2 \log(d_A)$.
		
		Consider for the moment, the $d_A^2$-outcome rank-one POVM defined by
		\begin{equation}\label{eq:uniform-povm}
			M_{d_A x+y}=\frac{1}{d_A^2}\,\outer{y}
		\end{equation}
		for $x,y\in [d_A]$, where $\{\ket{y}\}_y$ is an orthonormal basis that is mutually unbiased with the eigenbasis of $\rho_A$. This measurement is comprised of $d_A$ copies of the $\{\ket{y}\}_y$ basis measurement. Evidently, this measurement isn't extremal, but it is rank-one and it achieves the desired uniform distribution. By Lemma~\ref{lem:density-rank-one-povm}, we know that for any rank-one POVM $\{M_x\}_x$ with no more than $d_A^2$ outcomes we can always find an extremal rank-one POVM $\{N_x\}_x$ that is $\epsilon$-close to $\{M_x\}_x$ for any $\epsilon > 0$. Thus, by continuity of the mappings $\rho_{A} \to X'$, $\rho_A \to X''$ and $X \to H_\alpha(X)$ we can find a sequence of extremal rank-one POVMs $\{N_x^{(k)}\}_x$ that tends to~\eqref{eq:uniform-povm} and which establishes that
		\begin{equation}
			\sup_{\substack{\mathrm{Extremal~rank-one}\\ \{M_x\}_x}} \swHu{\alpha}(X|E) = 2 \log(d_A) - H_{\frac{\alpha}{2\alpha-1}}(A)
		\end{equation}
		and 
		\begin{equation}
			\sup_{\substack{\mathrm{Extremal~rank-one}\\ \{M_x\}_x}} \swHd{\alpha}(X|E) = 2 \log(d_A) - H_{1/\alpha}(A)\,.
		\end{equation}
	\end{proof}
\end{theorem}

 \section{Extension to other conditional entropies}
 The maximal intrinsic randomness problem formulated in Section~\ref{sec:intrinsic-rand-proj} and Section~\ref{sec:intrinsic-rand-povm} can be expressed for any conditional entropy. The choice of sandwiched Rényi entropies is natural due to their operational connections to randomness extraction which we demonstrate an application of in Section~\ref{sec:application}. Nevertheless there is nothing to stop one attempting to solve it for other conditional entropies such as the Petz-Rényi conditional entropies~\cite{Petz} or the Geometric entropies~\cite{M15}.
 
 However, despite the apparent simpler appearance of these entropies, the problem appears to be significantly more difficult. In particular, it is not a priori clear how to perform a reduction to rank-one measurements, which was a crucial step in our proof, due to the fact that these entropies do not satisfy a result akin to Lemma~\ref{lem:classical-info-ent}. As a consequence they do not agree with our intuitions regarding which measurements should produce maximal randomness. For example take the Petz-Rényi entropy
 \begin{equation}
 	\pHd{\alpha}(X|E)_{\rho_{XE}} = \frac{1}{1-\alpha} \log \tr{\rho_{XE}^{\alpha} \rho_E^{1-\alpha}}\,,
 \end{equation}
 which we define for $\alpha \in (1,2]$. Consider the setting of projective measurements only, we may assume that measuring in a basis mutually unbiased with the eigenbasis of $\rho_A$ would produce maximal randomness as it always produces a uniform distribution. However, for example let $\alpha = 3/2$ and consider the qutrit state $\rho = \tfrac14 \outer{0} + \tfrac34 \outer{1}$, then the unbiased basis measurement will give $\pHd{3/2}(X|E) = \log(3) - 2 \log((1+\sqrt{3})/2) \approx 0.685$. Whereas if we take the measurement defined by the orthonormal basis 
 \begin{equation}
 	\ket{m_0} = \ket{0} \qquad \ket{m_1} = \frac{\ket{1} + \ket{2}}{\sqrt{2}} \qquad \ket{m_2} = \frac{\ket{1} - \ket{2}}{\sqrt{2}}
 \end{equation}
 then we find $\pHd{3/2}(X|E) = 2 \log(\tfrac{4 \sqrt{2}}{\sqrt{2}+3}) \approx 0.716$. This shows that solving even $\maxintrinsicpvm{\HH}{\rho_A}$ may be very difficult for these other conditional entropies.

 \section{Additional lemmas}

\begin{lemma}\label{lem:classical-info-ent}
	Let $\rho_{XYE}$ be a state classical on $XY$. Then for $\alpha \in [1/2,1) \cup (1, \infty)$ we have 
	\begin{equation}
		\begin{aligned}
			\swHu{\alpha}(XY|E) &\geq \swHu{\alpha}(X|E) \\
			\swHd{\alpha}(XY|E) &\geq \swHd{\alpha}(X|E)\,.
		\end{aligned}
	\end{equation}
	\begin{proof}
		Writing $\rho_{XYE} = \sum_{xy} \outer{xy} \otimes \rho_{E,xy}$, if we take any density matrix $\sigma_E$ such that for all \(x,y\) we have $\supp(\rho_{E,xy}) \subseteq \sigma_E$, then
		\begin{equation}
			\swD{\alpha}(\rho_{XYE}\|\id_{XY} \otimes \sigma_E) = \frac{1}{\alpha-1} \log\left(\sum_{xy} \swQ{\alpha}(\rho_{E,xy} \| \sigma_E)\right)\,
		\end{equation}
		where $\swQ{\alpha}(\rho\|\sigma)\defed \tr{\left(\sigma^{\frac{1-\alpha}{2\alpha}} \rho \sigma^{\frac{1-\alpha}{2\alpha}}\right)^\alpha}$. For $\alpha > 1$, $\swQ{\alpha}$ is superadditive in the first argument~\cite{BCGM23}, in the sense that 
		\begin{equation}
			\swQ{\alpha}(\rho_0+\rho_1\|\sigma) \geq \swQ{\alpha}(\rho_0\|\sigma) + \swQ{\alpha}(\rho_1 \|\sigma)
		\end{equation}
		for any $\rho_0, \rho_1 \succeq 0$. Note for $\alpha \in [1/2,1)$ the inequality is reversed and $\swQ{\alpha}$ is subadditive. Thus for $\alpha>1$, using the monotonicity of the logarithm we have
		\begin{equation}
		\begin{aligned}
				\swD{\alpha}(\rho_{XYE}\|\id_{XY} \otimes \sigma_E) &\leq \frac{1}{\alpha-1}  \log\left(\sum_{x} \swQ{\alpha}\left( \sum_y \rho_{E,xy} \bigg\| \sigma_E\right)\right) \\
				&= \swD{\alpha}(\rho_{XE} \| \id_X \otimes \sigma_E)\,.
		\end{aligned}
		\end{equation}
		Taking $\sigma_E = \rho_E$ this immediately implies $\swHd{\alpha}(XY|E)  \geq  \swHd{\alpha}(X|E)$. For the optimized variant we have
		\begin{equation}
			\begin{aligned}
				\swHu{\alpha}(X|E) &= \sup_{\substack{\sigma_E \\ \sum_y\rho_{E,xy} \ll \sigma_E}} -\swD{\alpha}(\rho_{XE} \| \id_X \otimes \sigma_E) \\
				&\leq \sup_{\substack{\sigma_E \\ \sum_y\rho_{E,xy} \ll \sigma_E}} -\swD{\alpha}(\rho_{XYE} \| \id_{XY} \otimes \sigma_E) \\
				&\leq \sup_{\substack{\sigma_E \\ \rho_{E,xy} \ll \sigma_E}} -\swD{\alpha}(\rho_{XYE} \| \id_{XY} \otimes \sigma_E) \\
				&= \swHu{\alpha}(XY|E)\,.
			\end{aligned}
		\end{equation}
		where on the third line we have noted that 
		\begin{equation}
			\{\sigma_E \in \D(E) \, \mid \, \sum_y \rho_{E,xy} \ll \sigma_E \quad \forall x\} \subseteq 	\{\sigma_E \in \D(E) \, \mid \, \rho_{E,xy} \ll \sigma_E \quad \forall x,y\}\,.
		\end{equation}
		The proof for $\alpha < 1$ is the same but uses subadditivity of $\swQ{\alpha}$ instead. 
	\end{proof}
\end{lemma}
\begin{lemma}
    \label{lem:convexitypowerof2}
    Let be $\alpha \in [\tfrac12,1) \cup (1,\infty)$. Let \(\HH\) be either \(\swHd{\alpha}\) or \(\swHu{\alpha}\). Let $\rho_{XE}\in\D(XE)$, $\sigma_{XE} \in \D(XE)$, $\lambda \in [0,1]$ and define 
    \begin{equation}
    	\rho_{XEF} = \lambda \rho_{XE} \otimes \outer{0}_F + (1-\lambda) \sigma_{XE} \otimes \outer{1}_F\,.
    \end{equation}
    Then
    \begin{equation}
        2^{\HH(X|EF)_{\rho_{XEF}}} \leq \lambda2^{\HH(X|E)_{\rho_{XE}}}+(1-\lambda)2^{\HH(X|E)_{\sigma_{XE}}}.
    \end{equation}
\end{lemma}
\begin{proof}
    We have from \cite[Proposition~5.4]{Tomamichel_book}
    \begin{equation}
        2^{(1-\alpha)\swHd{\alpha}(X|EF)}=\lambda\,2^{(1-\alpha)\swHd{\alpha}(X|E)_{\rho_{XE}}}+(1-\lambda)\,2^{(1-\alpha)\swHd{\alpha}(X|E)_{\sigma_{XE}}}\,.
    \end{equation}
    Now, note that \(x\mapsto x^{\frac{1}{1-\alpha}}\) is convex for \(\alpha \in [\tfrac12,1) \cup (1,\infty)\). Thus, by Jensen's inequality we have the desired result for \(\swHd{\alpha}\).

    Similarly, from the same proposition, we have
    \begin{equation}
        2^{\frac{1-\alpha}{\alpha}\swHu{\alpha}(X|EF)}=\lambda\,2^{\frac{1-\alpha}{\alpha}\swHu{\alpha}(X|E)_{\rho_{XE}}}+(1-\lambda)\,2^{\frac{1-\alpha}{\alpha}\swHu{\alpha}(X|E)_{\sigma_{XE}}}\,.
    \end{equation}
    Since \(x\mapsto x^{\frac{\alpha}{1-\alpha}}\) is convex for \(\alpha \in [\tfrac12,1) \cup (1,\infty)\), we once again have the desired result by Jensen's inequality.
\end{proof}

\begin{lemma}\label{lem:extremal-to-rankone}
	Let $\{M_x\}_x$ be an extremal POVM and let 
	\begin{equation}
		M_x=\sum_y\lambda_{x, y}\,\outer{\varphi_{x, y}}\,.
	\end{equation}
	be its spectral decomposition. Then $\{\lambda_{x,y} \outer{\varphi_{x,y}}\}_{xy}$ is an extremal rank-one POVM where the indexes range over $(x,y)$ such that $\lambda_{x,y} > 0$. 
\end{lemma}
\begin{proof}
	It is clear that \(\left\{\lambda_{x, y}\outer{\varphi_{x,y}}\right\}_{x,y}\) is a POVM, since all its elements are positive and we have
	\begin{equation}
		\sum_{x, y}\lambda_{x, y}\,\outer{\varphi_{x, y}}=\sum_xM_x=\id\,.
	\end{equation}
	Thus it is a rank-one POVM. We now show that it inherits extremality from $\{M_x\}_x$. Since a rank-one POVM is extremal if and only if its elements are linearly independent~\cite[Corollary~2.50]{watrous}, let us assume for contradiction that there exist complex coefficients \(\left\{\alpha_{x, y}\right\}_{x, y}\) (some of which are non-zero) such that
	\begin{equation}
		\sum_{x, y}\alpha_{x, y}\lambda_{x, y}\,\outer{\varphi_{x, y}}=0\,,
	\end{equation}
	i.e., the POVM elements $N_{x,y} = \lambda_{x, y} \outer{\varphi_{x, y}}$ are linearly dependent.
	For each $x$, define \(\theta_x\) as
	\begin{equation}
		\theta_x \defed \sum_y\alpha_{x, y}\lambda_{x, y}\,\outer{\varphi_{x, y}}
	\end{equation}
	and define
	\begin{equation}
		\tau_x \defed \theta_x + \theta_x^\dagger\,.
	\end{equation}
	Clearly, for all \(x\), \(\tau_x\) is Hermitian. Furthermore, by assumption we have
	\begin{equation}
		\sum_x\tau_x=0
	\end{equation}
	and moreoever by construction we have for each $x$ that
	\begin{equation}
		\supp\left(\tau_x\right)\subseteq\supp\left(M_x\right)\,.
	\end{equation}
	Since \(\left\{M_x\right\}_x\) is extremal, it follows from~\cite[Theorem~2.47]{watrous} that we must have $\tau_x=0$ for each $x$. Thus, \(\mathrm{i}\theta_x\) is Hermitian and also satisfies the aforementioned condition, which gives us \(\theta_x=0\) for all \(x\), i.e.,
	\begin{equation}
		\sum_y \alpha_{x,y} \lambda_{x, y} \outer{\varphi_{x,y}} = 0\,.
	\end{equation}
	Due to orthogonality, this implies that $\alpha_{x,y} \lambda_{x, y}=0$ for all $x,y$, which implies $\alpha_{x, y} = 0$. This is in contradiction with the original assumption of linear dependence and hence the measurement $\{N_{x,y}\}_{x,y}$ is extremal.
\end{proof}

\begin{lemma}[Density of extremal rank-one POVMs]
    \label{lem:density-rank-one-povm}
     Let $A$ be a quantum system of dimension $d$ and let \(N\leq d^2\). Let \(\left\{\outer{m_x}\right\}_{x}\) be rank-one POVM on $A$ with \(N\) outcomes. Then for all \(\varepsilon>0\), there exists an extremal rank-one POVM \(\left\{\outer{e_x}\right\}_{x}\) with $N$ outcomes such that
     \begin{equation}
         \left\|\outer{m_x}-\outer{e_x}\right\|\leq\varepsilon\,
     \end{equation}
     for all $x$.
 \end{lemma}
 \begin{proof}
     If $\{\outer{m_x}\}_x$ is extremal then the statement is immediately true. Suppose thus that the measurement is not extremal. This means that its elements are linearly dependent when viewed as elements of the $\RR$-vector space \(\mathrm{Herm}\left(\CC^d\right)\)~\cite[Corollary~2.50]{watrous}. That is, there exist some real coefficients \(\alpha_1,\cdots,\alpha_{N}\) that are not all zero such that
     \begin{equation}
         \sum_{x=0}^{N-1}\alpha_x\,\outer{m_x}=0\,.
     \end{equation}
     Let \(S\) denote the real span of \(\left\{\outer{m_x}\right\}_x\). Since \(\left\{\outer{m_x}\right\}_x\) is linearly dependent, we have \(\dim(S)<N\leq d^2\), with \(d^2\) being the dimension of \(\mathrm{Herm}\left(\CC^d\right)\) viewed as a \(\RR\)-vector space. In particular, this means we can find a subspace \(E\) orthogonal to $S$ with non-zero dimension such that
     \begin{equation}
         \mathrm{Herm}\left(\CC^d\right) = S \oplus E
     \end{equation}
     where the orthogonality is defined using the Hilbert-Schmidt inner product $\langle X, Y\rangle = \tr{X Y}$. Now consider any $X \in E$, by the spectral theorem $X= \sum_i \lambda_i \outer{w_i}$ with $\lambda_i \in \RR$. Thus at least one of the $\outer{w_i} \not\in S$ otherwise $X \in S$. Let us denote by \(\outer{w}\) such a projector and let us decompose it as
     \begin{equation}
         \outer{w} = Y_S + Y_E
     \end{equation}
     with \(Y_S\in S\) and \(Y_E\in E\), by assumption $Y_E \neq 0$. 
     
     Now, without loss of generality, let us assume that \(\alpha_{N-1}\neq0\), which means that \(\outer{m_{N-1}}\in S\). For any \(\varepsilon\in(0, 1)\), let us define
     \begin{equation}
         \ket{\tau_\varepsilon}\defed\sqrt{1-\varepsilon}\,\ket{m_{N-1}} + \sqrt{\varepsilon}\,\ket{w}\,.
     \end{equation}
     We want to show that \(\outer{\tau_\varepsilon}\notin S\), so that we can replace \(\outer{m_{N-1}}\) by \(\outer{\tau_\varepsilon}\) to increase the dimension of the span of the measurement. For the sake of the contradiction, let us assume that \(\outer{\tau_{\varepsilon}}\in S\). We then have
     \begin{equation}
         0=\tr{\outer{\tau_\varepsilon}Y_E}=\varepsilon\tr{Y_E^2}+\sqrt{\varepsilon(1-\varepsilon)}\tr{Y_E\left(\ketbra{w}{m_{N-1}}+\ketbra{m_{N-1}}{w}}\right)\,.
     \end{equation}
     Since \(\tr{Y_E^2}>0\), there is at most a single \(\varepsilon\in(0, 1)\) such that this equation is satisfied. Hence, it is always possible to find an arbitrary small \(\varepsilon\in(0, 1)\) such that \(\outer{\tau_{\varepsilon}}\notin S\). Thus, by replacing \(\outer{m_{N-1}}\) by \(\outer{\tau_\varepsilon}\), the dimension of $S$ increases by 1.

     To make sure the new collection of operators forms a POVM we may need to renormalize. In particular, define
     \begin{equation}
         G_\varepsilon = \outer{\tau_\varepsilon} + \sum_{x=0}^{N-2}\outer{m_x} 
     \end{equation}
     and then consider, for \(x\in[N-1]\)
     \begin{equation}
         \outer{m_x'}\defed G_\varepsilon^{-\frac12}\outer{m_x}G_\varepsilon^{-\frac12}
     \end{equation}
     and \(\outer{m_{N-1}'}=G_\varepsilon^{-\frac12}\outer{\tau_\varepsilon}G_\varepsilon^{-\frac12}\). Such a construction yields a rank-one POVM whose span has dimension \(\dim(S)+1\). We now show that we can make this new POVM arbitrarily close to the original POVM, first note that
     \begin{equation}
         \begin{aligned}
             G_\varepsilon - \id_d &= \outer{\tau_\varepsilon} - \outer{m_{N-1}}\\
             &= \varepsilon\left(\outer{w}-\outer{m_{N-1}}\right)+\sqrt{\varepsilon(1-\varepsilon)}\left(\ketbra{m_{N-1}}{w}+\ketbra{w}{m_{N-1}}\right)
         \end{aligned}
     \end{equation}
     and thus $\lim_{\varepsilon \to 0} G_{\varepsilon} = \id_d$. Now, by continuity of \(M\mapsto M^{-\frac12}\) on some open ball with center \(\id_d\), we also have that \(G_\varepsilon^{-\frac12}\) converges to \(\id_d\) as \(\varepsilon\) goes to 0. In particular, for all \(x \in [N]\), we have that \(G_\varepsilon^{-\frac12}\outer{m_x'}G_\varepsilon^{-\frac12}\) converges to \(\outer{m_x}\) as \(\varepsilon\) goes to 0.

     Thus, we managed to find a rank-one POVM such that the dimension of the span of its elements is strictly larger than the previous one, while being arbitrarily close to the previous one. By iterating this process, we can finally find an extremal rank-one POVM arbitrarily close to the original one.
 \end{proof}
\end{document}